\documentclass[a4paper, 10pt,psamsfonts]{amsart}
\usepackage{amsmath}
\usepackage{amsthm}
\usepackage{amssymb}
\usepackage{amscd}
\usepackage{amsfonts}
\usepackage{amsbsy}

\def\be{\begin{equation}}
\def\ee{\end{equation}}

\def\bq{\begin{eqnarray}}
\def\eq{\end{eqnarray}}

\def\beq{\begin{eqnarray*}}
\def\eeq{\end{eqnarray*}}

\newtheorem {theorem} {Theorem}
\newtheorem {proposition} [theorem]{Proposition}
\newtheorem {corollary} [theorem]{Corollary}
\newtheorem {lemma}  [theorem]{Lemma}

\newtheorem {remark} [theorem]{Remark}
\newtheorem {conjecture} [theorem]{Conjecture}

\newcommand{\R}{\mathbb{R}}

\newcommand{\MM}{\mathcal{M}}
\begin{document}

\title[A new approach  to the vakonomic mechanics ]
{A new approach  to the vakonomic mechanics}

\date{}
\dedicatory{}
\author[J. Llibre, R. Ram\'{\i}rez and N. Sadovskaia]
{Jaume Llibre$^1$, Rafael Ram\'{\i}rez$^2$ and Natalia
Sadovskaia$^3$}

\address{$^1$ Departament de Matem\`{a}tiques,
Universitat Aut\`{o}noma de Barcelona, 08193 Bellaterra,
Barcelona, Catalonia, Spain.} \email{jllibre@mat.uab.cat}

\address{$^2$ Departament d'Enginyeria Inform\`{a}tica i Matem\`{a}tiques,
Universitat Rovira i Virgili, Avinguda dels Pa\"{\i}sos Catalans
26, 43007 Tarragona, Catalonia, Spain.}
\email{rafaelorlando.ramirez@urv.cat}

\address{$^3$ Departament de Matem\`{a}tica Aplicada II, Universitat
Polit\`{e}cnica de Catalunya, C. Pau Gargallo 5, 08028 Barcelona,
Catalonia, Spain.} \email{natalia.sadovskaia@upc.edu}



\maketitle
\begin{abstract}
The aim of this paper is to show that the Lagrange--d'Alembert and
its equivalent the Gauss and Appel principle are not the only way to
deduce the equations of motion of the nonholonomic systems. Instead
of them, here we consider  the generalization of  the Hamiltonian
principle for nonholonomic systems with nonzero transpositional
relations.

 By applying this variational principle which takes into the account transpositional relations different
  from the classical ones we deduce the
equations of motion for the nonholonomic systems with constraints
that in general are nonlinear in the velocity. These equations of
motion coincide, except perhaps in a zero Lebesgue measure set, with
the classical differential equations deduced with
d'Alembert--Lagrange principle.

\smallskip

 We provide a new point of view on the transpositional relations for the constrained mechanical
systems: the  virtual variations
 can produce zero or non--zero transpositional relations. In
 particular the independent virtual variations can produce  non--zero transpositional
 relations.   For the unconstrained  mechanical systems the virtual
 variations always produce zero transpositional relations.

\smallskip

We conjecture that the existence of the nonlinear constraints in the
velocity must be sought outside of the Newtonian model.

\smallskip

 All our results are illustrated with precise examples.

\end{abstract}

\section{Introduction}

The history of nonholonomic mechanical systems is long and complex
and goes back to the 19 century, with important contribution by
Hertz \cite{Hertz} (1894) , Ferrers \cite{Ferrers} (1871),  Vierkandt \cite{Vierkandt} (1892) and
Chaplygin \cite{Chaplygin} (1897).

The nonholonomic mechanic is a remarkable generalization  of the
classical Lagrangian and Hamiltonian mechanic. The birth of the
theory of dynamics of nonholonomic systems occurred when
Lagrangian-Euler formalism was found  to be inapplicable for
studying the  simple mechanical problem of a rigid body rolling
without slipping on a plane.

\smallskip

 A long period of time has been needed for finding the
correct equations of motion of the nonholonomic mechanical systems
and the study of the deeper questions associated with the geometry
and the analysis of these equations. In particular the integration
theory of equations of motion for nonholonomic mechanical systems
is not so complete as in the case of holonomic systems. This is
due to several reasons. First, the equations of motion of
nonholonomic systems have more complex structure than the Lagrange
one, which describes the behavior of holonomic systems. Indeed, a
holonomic systems can be described by a unique function of its
state and time, the Lagrangian function. For the nonholonomic
systems this is not possible.  Second, the equations of motion of
nonholonomic systems in general have no invariant measure, as they
have the equations of motion of holonomic systems (see
\cite{Kozlov11,Leon,LlRS111,Vershik}).

\smallskip

 One of the most important
directions in the development of the nonholonomic mechanics  is
the research   connected with the general mathematical formalism
to describe the behavior of such systems which differs from the
Lagrangian and Hamiltonian formalism. The main problem with the
equations of motion of the nonholonomic mechanics has been
centered on whether or not these equations can be derived from the
Hamiltonian principle  in the usual sense, such as for the
holonomic systems (see for instance \cite{Marle}). But there is
not doubt that the correct equations of motion for nonholonomic
systems are given by the d'Alembert--Lagrange principle.

The general understanding of inapplicability of Lagrange equations
and variational Hamiltonian principles to the nonholonomic systems
is due to Hertz, who expressed it in his fundamental work {\it Die
Prinzipien der Mechanik in neuem Zusammenhaange dargestellt}
\cite{Hertz}.  Hertz's ideas were developed by Poincar{\'e} in
\cite{Poincare}. At the same time various aspects of nonholonomic
systems need to be studied such as

\smallskip

(a) The problem of the realization of nonholonomic constraints (see
for instance \cite{Kozlov10, Kozlov001}).

\smallskip

(b) The stability of nonholonomic systems (see for instance \cite{
NF,Ruban}).

\smallskip

(c)  The role of the so called {\it transpositional relations} (see
\cite{Kirgetov,Marushin,NF,Ramirez})
 \begin{equation}\label{13}
\delta\dfrac{d\textbf{x}}{dt}-\dfrac{d}{dt}\delta{\textbf{x}}=\left(\delta\dfrac{dx_1}{dt}-
\dfrac{d}{dt}\delta{x_1},\ldots,\delta\dfrac{dx_N}{dt}-\dfrac{d}{dt}\delta{x_N}\right),
\end{equation}
 {where
$\dfrac{d}{dt}$ denotes the differentiation with respect to the
time, $\delta$ is the virtual variation, and
$\textbf{x}=\left(x_1,\ldots,x_N\right)$ is the vector of the
generalized coordinates}.

{Indeed the most general formulation of the Hamiltonian}
principle is the {\it{Hamilton--Suslov principle}}
\begin{equation}\label{Su11}
\displaystyle\int_{t_0}^{t_1}\left(\delta\,{\tilde{L}}-
\displaystyle\sum_{j=1}^N\dfrac{\partial{{\tilde{L}}}}{\partial\dot{x}_j}
\left(\delta\dfrac{dx_j}{dt}-\dfrac{d}{dt}\delta{x_j}\right)\right)dt=0,
\end{equation}
suitable for
constrained and unscontrained Lagrangian systems, where $\tilde{L}$
is the Lagrangian of the mechanical system.. Clearly the
equations of motion obtained from the Hamilton--Suslov principle
depend on the point of view on the transpositional relations. This
fact shows the importance of these relations.

\smallskip

(d)  The relation  between nonholonomic
mechanical systems and {\it vakonomic mechanical systems}.

There was some confusion in the literature between nonholonomic
mechanical systems and variational nonholonomic mechanical systems
also called { vakonomic mechanical systems}. Both kinds of systems
have the same mathematical ``ingredients": a Lagrangian function
and a set of constraints. But the way in which the equations of
motion are derived differs. As we observe the equations of motion
in nonholonomic mechanic are deduced using d'Alembert--Lagrange's
principle. In the case of vakonomic mechanics the equations of
motion are obtained through the application of a constrained
variational principle. The term vakonomic (``variational axiomatic
kind") is due to Kozlov (see \cite{Kozlov1,Kozlov2,Kozlov3}), who
proposed this mechanics as an alternative set of equations of
motion for a constrained Lagrangian systems.

The distinction between the classical differential equations of
motion  and the equations of motion of variational nonholonomic
mechanical systems has a long history going back to the survey
article of Korteweg  (1899) \cite{Korteweg} and discussed in a
more modern context in \cite{Favretti,Kharlamov, Lewis,Zampiery}.
In these papers the authors have discussed the domain of the
vakonomic  and nonholonomic mechanics.
 In the paper {\it{Critics  of some mathematical model to describe the behavior of mechanical
 systems with differential constraints}}  \cite{Kharlamov},
 Kharlamov studied the Kozlov model and
 in a concrete example showed that the subset
 of solutions of the studied nonholonomic
 systems is not included in the set of vakonomic model and proved that the
 principle of determinacy is not valid in the Kozlov
 model. In \cite{Kupka} the authors put in evidence the main differences
 between the d'Alembertian and the vakonomic approaches.  From the results obtained in
 several papers it follows that in general the
 vakonomic model is not applicable to the nonholonomic constrained
 Lagrangian systems.

 The equations of motion for the constrained mechanical systems deduced
by  Kozlov (see for instance \cite{Arnold}) from the Hamiltonian
principle with the Lagrangian
  $L:\mathbb{R}\times{T}\textsc{Q}\times\mathbb{R}^M\longrightarrow\mathbb{R}$
  such that
 $L=L_0-\displaystyle\sum_{j=1}^M\lambda_jL_j,$
where $L_j=0$ for $j=1,\ldots,M<N$ are the given constraints, and
$L_0$ is the classical Lagrangian. These equations are
\begin{equation}
\label{K1}
E_kL={\dfrac{d}{dt}\dfrac{\partial
L}{\partial\dot{x}_k}-\dfrac{\partial
L}{\partial{x}_k}}=0\Longleftrightarrow
E_kL_0=\displaystyle\sum_{j=1}^M\left(\lambda_jE_k\,L_j+
\dfrac{d\lambda_j}{dt}\dfrac{\partial{L_j}}{\partial
\dot{x}_k}\right),
\end{equation}
for $k=1,\ldots,N,$ see for more details \cite{Arnold}. Clearly,
equations \eqref{K1} differ from the classical equations  by the
presence of the terms $\lambda_jE_k\,L_j.$ If the constraints are
integrable, i.e. $L_j=\dfrac{d}{dt}g_j(t,\textbf{x}),$  then the
vakonomic mechanics reduces to the holonomic one.

\smallskip

 In this paper we give a modification of the vakonomic
 mechanics. This modification is valid for the
holonomic and nonholonomic constrained Lagrangian systems. We
apply the generalized constrained Hamiltonian principle with
non--zero transpositional relations.  By applying this constrained variational principle we
deduce the equations of motion for the nonholonomic systems with
constraints which in general are nonlinear in the velocity. These
equations coincide, except perhaps in a zero Lebesgue measure set,
with the classical differential equations deduced from
d'Alembert--Lagrange principle.

\section{Statement of the main results}
 In this paper we solve the following {\it inverse problem of the constrained
Lagrangian systems} (see \cite{LlRS222})

\smallskip

We consider the constrained Lagrangian systems with configuration
space $\textsc{Q}$ and phase space $T\textsc{Q}.$

\smallskip

Let
$
L:\mathbb{R}\times{T\textsc{Q}}\times{\mathbb{R}^M}
\longrightarrow\mathbb{R}
 $
  be a smooth function such that
\begin{equation}\label{21}
{L}\left(t,\textbf{x},\dot{\textbf{x}},\Lambda
\right)=L_0\left(t,\textbf{x},\dot{\textbf{x}}\right)-
\displaystyle\sum_{j=1}^M\lambda_j\,L_j\left(t,\textbf{x},\dot{\textbf{x}}\right)-\displaystyle\sum_{j=M+1}^N
\lambda^0_j L_j\left(t,\textbf{x},\dot{\textbf{x}}
\right),
\end{equation}
where $\Lambda=\left(\lambda_1,\ldots,\lambda_M\right)$ are the
additional coordinates (Lagrange multipliers),
$
L_j:\mathbb{R}\times{T\textsc{Q}}\longrightarrow
\mathbb{R},\quad \left(t,\textbf{x},\dot{\textbf{x}}\right)\longmapsto
\,L_j\left(t,\textbf{x},\dot{\textbf{x}}\right),
$ be smooth functions for $j=0,\ldots,N,$ where $L_0$ is the nonsingular function i.e.
$\det\left(\dfrac{\partial^2
L_0}{\partial{\dot{x}_k}\partial{\dot{x}_j}}\right)\ne{0},$
and
$L_j=0,$ for
$j=1,\ldots,M,$ are the  constraints  satisfying
\begin{equation}\label{TT}
\mbox{rank}\left(\dfrac{\partial(L_1,\ldots,L_M)}
{\partial(\dot{x}_1,\ldots,\dot{x}_N)} \right)=M
\end{equation}
 in all the
points of $ \mathbb{R}\times T\textsc{Q},$ except perhaps in a zero
Lebesgue measure set, $L_j$ and $ \lambda^0_j$ are arbitrary
functions and  constants respectively, for $ j=M+1,\ldots,N$.

\smallskip

 We must determine
the smooth functions
$L_j,$  constants $\lambda^0_j$ for $j=M+1,\ldots,N$ and the matrix $A$ in
such a way that the differential equations describing the behavior
of the constrained Lagrangian systems and obtained from the
 the Hamiltonian principle
\begin{equation}
\label{023}
 \displaystyle
\displaystyle\int_{t_0}^{t_1}\delta\,{L}=
\displaystyle\int_{t_0}^{t_1}\left(\dfrac{\partial
L}{\partial x_j}\delta x_j+\dfrac{\partial L}{\partial
\dot{x}_j} \dfrac{d}{dt}{\delta x}_j+
\displaystyle\sum_{j=1}^N\dfrac{\partial{{L}}}{\partial\dot{x}_j}
\left(\delta\dfrac{dx_j}{dt}-\dfrac{d}{dt}\delta{x_j}\right)\right)dt=0,
\end{equation}
 with transpositional relation given by
 \begin{equation} \label{23}
\delta\dfrac{d\textbf{x}}{dt}-\dfrac{d}{dt}\delta{\textbf{x}}=
A\left(t,\textbf{x},\dot{\textbf{x}},\ddot{\textbf{x}}
\right)\delta{\textbf{x}},
\end{equation}
 where
$A=A\left(t,\textbf{x},\dot{\textbf{x}},\ddot{\textbf{x}}\right)=
\left(A_{\nu\,j}\left(t,\textbf{x},\dot{\textbf{x}},\ddot{\textbf{x}}\right)\right)$
is a $N\times N$  matrix,

\smallskip

We give the solutions of this problem in two steps. First we
obtain the differential equations along the solutions satisfying
\eqref{023}. Second we shall contrast the obtained equations and
classical differential equations which described the  behavior of
the constrained mechanical systems. The solution of this inverse
problem is presented in section 4.

\smallskip

Note that the function $L$ is singular, due to the absence of
$\dot{\lambda}.$

\smallskip
{We observe that} the arbitrariness of the functions $L_j,$ of the
constants $\lambda^0_j$ for $j=M+1,\ldots,N,$ and of the matrix
$A$ will play a fundamental role in the construction of the
mathematical model which we propose in this paper.

\smallskip

Our main results are the following

\begin{theorem}\label{A}
 We assume that $\delta{x_\nu(t)},\quad \nu=1,\ldots, N,$ are arbitrary
  functions defined in the interval $[t_0,\,t_1]$, smooth in the interior
  of $[t_0,\,t_1]$ and  vanishing at its
endpoints, i.e., $\delta{x_\nu}({t_0})=\delta{x_\nu}({t_1})=0.$ If
\eqref{23} holds
 then the path $
\gamma(t)=(x_1(t), \ldots, x_N(t))$ compatible with the constraints
$L_j\left(t,\textbf{x},\dot{\textbf{x}}\right)=0$, for
$j=1,\ldots,M$  satisfies \eqref{023}{ with $L$ given by the formula
}\eqref{21} if and only if it is a solution of the differential
equations
\begin{equation}\label{24}
 D_{\nu}L:= {E_\nu\,L-\displaystyle\sum_{j=1}^N{A_{{\nu}j}
\dfrac{\partial{L}}{\partial{\dot{x}_j}}}}=0,\quad \dfrac{\partial
L}{\partial\lambda_k}=-L_k=0,
\end{equation}
 for $\nu=1,\ldots,N,$ and $k=1,\ldots,M,$ where
 $E_\nu={\dfrac{d}{dt}\dfrac{\partial
}{\partial\dot{x}_\nu}-\dfrac{\partial }{\partial{x}_\nu}}.$ {
System \eqref{24} is  equivalent to the following two differential
 systems}
\begin{equation}
\label{D24}
\begin{array}{rl}
 D_{\nu}L_0=&\displaystyle\sum_{j=1}^M\left(\lambda_{j}D_{\nu}L_{j}+
\dfrac{d\lambda_j}{dt}\dfrac{\partial{L_j}}{\partial{\dot{x}_\nu}}\right)
+\displaystyle\sum_{j=M+1}^N \lambda^0_j D_\nu\,L_j,\quad
L_k=0\Longleftrightarrow\vspace{0.2cm}\\
E_{\nu}L_0=&\displaystyle\sum_{k=1}^NA_{jk}\dfrac{\partial
L_0}{\partial\dot{x}_k}+\sum_{j=1}^M\left(\lambda_{j}D_{\nu}L_{j}+
\dfrac{d\lambda_j}{dt}\dfrac{\partial{L_j}}{\partial{\dot{x}_\nu}}\right)
+\displaystyle\sum_{j=M+1}^N \lambda^0_j D_\nu\,L_j,\quad L_k=0.
\end{array}
\end{equation}
 for $\nu=1,\ldots,N$ and $k=1,\ldots,M.$
\end{theorem}
\begin{theorem}\label{A1}
Using the notation of Theorem \ref{A} let
\begin{equation}\label{b1}
L=L\left(t,\textbf{x},\dot{\textbf{x}},\Lambda\right)=L_0\left(t,\textbf{x},\dot{\textbf{x}}\right)-
\displaystyle\sum_{j=1}^M\lambda_j\,L_j\left(t,\textbf{x},
\dot{\textbf{x}}\right)-\displaystyle\sum_{j=M+1}^N \lambda^0_j
L_j\left(t,\textbf{x}, \dot{\textbf{x}}\right) \end{equation} be
the Lagrangian and let
$L_j\left(t,\textbf{x},\dot{\textbf{x}}\right)=0$ be the
independent constraints  for $j=1,\ldots,M<N,$ and let
$\lambda^0_k$ {be the arbitrary constants  for} $k=M+1,\ldots,N,$
$L_k:\mathbb{R}\times\,T\textsc{Q}\longrightarrow \mathbb{R}$ for
$k=M+1,\ldots,N$ arbitrary functions such that
\[
|W_1|=\det{W_1}=\det{\left(\dfrac{\partial(L_1,\ldots,L_N)}
{\partial(\dot{x}_1,\ldots,\dot{x}_N)} \right)}\ne 0,
\]
except perhaps in a zero Lebesgue measure set $|W_1|=0$. We
determine the matrix $A$ satisfying
 \begin{equation}\label{a1}
 W_1A=\Omega_1:=\left(\begin{array}{ccc}
E_1L_1&\hdots&E_NL_1\\
\vdots&\hdots&\vdots\\
\vdots&\hdots&\vdots\\
E_1L_{N}&\hdots&E_NL_{N}\\
\end{array}\right).
\end{equation}
Then the differential equations \eqref{D24} become
\begin{equation}
\label{210}
\begin{array}{rl}
D_{\nu}L_0=&\displaystyle\sum_{\alpha=1}^M\dot{\lambda}_\alpha\dfrac{\partial{L_\alpha}}
{\partial{\dot{x}_\nu}}\quad\mbox{for}\quad
\nu=1,\ldots,N \vspace{0.30cm}\\
\Longleftrightarrow &\dfrac{d}{dt}\dfrac{\partial L_0
}{\partial\dot{\textbf{x}}}-\dfrac{\partial L_0
}{\partial{\textbf{x}}}=\left(W^{-1}_1\Omega_1\right)^T
\dfrac{\partial{L_0}}{\partial{\dot{\textbf{x}}}}+W^T_1\dfrac{d\lambda}{dt},
\end{array}
\end{equation}
where  $\dfrac{\partial
}{\partial\dot{\textbf{x}}}=\left(\dfrac{\partial
}{\partial\dot{x_1}},\ldots,\dfrac{\partial
}{\partial\dot{x_N}}\right)^T,\,\,\dfrac{\partial
}{\partial{\textbf{x}}}=\left(\dfrac{\partial
}{\partial{x_1}},\ldots,\dfrac{\partial }{\partial{x_N}}\right)^T,$
$\lambda=\left(\lambda_1,\ldots,\lambda_M,0,\ldots,0\right)^T,$ and
the transpositional relation \eqref{23} becomes
\begin{equation}
\label{211}
\delta\dfrac{d\textbf{x}}{dt}-\dfrac{d}{dt}\delta{\textbf{x}}=
\left(W^{-1}_1\Omega_1\right)\delta{\textbf{x}}.
\end{equation}
\end{theorem}

\begin{theorem}\label{A2}
Using the notation of Theorem \ref{A} let
\begin{equation}\label{b2}
L\left(t,\textbf{x},\dot{\textbf{x}},\Lambda\right)=
L_0\left(t,\textbf{x},\dot{\textbf{x}}\right)-
\displaystyle\sum_{j=1}^M\lambda_j\,L_j\left(t,\textbf{x},
\dot{\textbf{x}}\right)-\displaystyle\sum_{j=M+1}^{N-1}
\lambda^0_j L_j\left(t,\textbf{x}, \dot{\textbf{x}}\right)
\end{equation} be the Lagrangian and
$L_j\left(t,\textbf{x},\dot{\textbf{x}}\right)=0$ be the
independent constraints for $j=1,\ldots,M<N,$ and let
 $\lambda^0_j$ {be arbitrary constants,
for} $j=M+1,\ldots,N-1$ and $\lambda^0_N=0,$
$L_j:\mathbb{R}\times\,T\textsc{Q}\longrightarrow \mathbb{R}$ for
$j=M+1,\ldots,N-1$ arbitrary functions,   and $L_N=L_0$ such that
\[
|W_2|=\det{W_2}=\det{\left(\dfrac{\partial(L_1,\ldots,L_{N-1},L_0)}
{\partial(\dot{x}_1,\ldots,\dot{x}_N)} \right)}\ne 0,
\]
except perhaps in a zero Lebesgue measure set $|W_2|=0$. We
determine the matrix $A$ satisfying
 \begin{equation}\label{a2}
 W_2A=\Omega_2:=\left(\begin{array}{ccc}
E_1L_1&\hdots&E_NL_1\\
\vdots&\hdots&\vdots\\
E_1L_{N-1}&\hdots&E_NL_{N-1}\\
0&\hdots&0\\
\end{array}\right).
\end{equation}
Then the differential equations \eqref{D24} become
\begin{equation}
\label{0210}
 \dfrac{d}{dt}\dfrac{\partial L_0
}{\partial\dot{\textbf{x}}}-\dfrac{\partial L_0
}{\partial{\textbf{x}}}=W^T_2\dfrac{d}{dt}\tilde{\lambda},
\end{equation}
 where
$\lambda:=
\tilde{\lambda}=\left(\tilde{\lambda}_1,\ldots,\tilde{\lambda}_M,\,
0,\ldots,0\right)^T, $ and the transpositional relation \eqref{23}
becomes
\begin{equation}
\label{0211}
\delta\dfrac{d\textbf{x}}{dt}-\dfrac{d}{dt}\delta{\textbf{x}}=
\left(W^{-1}_2\Omega_2\right)\delta{\textbf{x}},
\end{equation}
\end{theorem}
The proofs  of Theorems \ref{A},\, \ref{A1} and \ref{A2} are given
in section 5.

\smallskip

\begin{theorem}\label{A0}
Under the assumptions of Theorem \ref{A1} and assuming that
 \[\begin{array}{rl}
 x_\alpha=&x_\alpha,\quad
x_{\beta}=y_\beta\quad \textbf{x}=\left(x_1,\ldots,x_{s_1}\right)\quad
\textbf{y}=\left(y_1,\ldots,y_{s_2}\right),\vspace{0.2cm}\\
L_\alpha=&\dot{x}_\alpha-\Phi_\alpha\left(\textbf{x},\textbf{y},\dot{\textbf{x}},\,\dot{\textbf{y}}\right)=0,\quad
L_\beta=\dot{y}_\beta,
\end{array}
\]
{for} $\alpha=1,\ldots,s_1=M$ and $\beta=s_1+1,\ldots,s_1+s_2=N.$

\smallskip

Then $|W_1|=1$ and the differential equations \eqref{210} take the form
\begin{equation}\label{48}
\begin{array}{rl}
 E_jL_0=&\displaystyle\sum_{\alpha=1}^{s_1}\left(E_j L_\alpha \dfrac{\partial
L_{0}}{\partial\dot{x}_\alpha}\right)+\dot{\lambda}_j\quad
j=1,\ldots,s_1,\\
 E_{k}L_0=&\displaystyle\sum_{\alpha=1}^{s_1}\left(E_k L_\alpha\, \dfrac{\partial
L_{0}}{\partial\dot{x}_\alpha}+\dot{\lambda}_\alpha
\dfrac{\partial L_\alpha}{\partial\dot{y}_k}\right)\quad
k=1,\ldots,s_2.
\end{array}
\end{equation}
or, equivalently (excluding the Lagrange multipliers)
\begin{equation}\label{Pa3}
 E_{k}L_0=\displaystyle\sum_{\alpha=1}^{s_1}\left(E_k L_\alpha\, \dfrac{\partial
L_{0}}{\partial\dot{x}_\alpha}+\left(
 E_\alpha L_0-\displaystyle\sum_{\beta=1}^{s_1}\left(E_\alpha L_\beta \dfrac{\partial
L_{0}}{\partial\dot{x}_\beta}\right)\right)
\dfrac{\partial\,L_\alpha}{\partial\dot{y}_k}\right),\quad
k=1,\ldots,s_2.
\end{equation}
In particular if we choose
$L_0=\tilde{L}\left(\textbf{x},\textbf{y},\dot{\textbf{x}},\dot{\textbf{y}}\right)-
\tilde{L}\left(\textbf{x},\textbf{y},\Phi,\dot{\textbf{y}}\right)
=\tilde{L}-L^*,$ where $
\Phi=\left(\Phi_1,\ldots,\Phi_{s_1}\right),$  then \eqref{Pa3} holds
if
\[ E_{k}\tilde{L}=\displaystyle\sum_{\alpha=1}^{s_1}E_\alpha\tilde{L}
\dfrac{\partial\,L_\alpha}{\partial\dot{y}_k},\quad
k=1,\ldots,s_2,
\]
and
\begin{equation}\label{C48}
E_{k}(L^*)=\displaystyle\sum_{\alpha=1}^{s_1}\left(\dfrac{d}{dt}
\left(\dfrac{\partial\Phi_\alpha}{\partial\dot{y}_k}\right)-
\left(\dfrac{\partial{\Phi_\alpha}}{\partial\,{y}_k}+
\displaystyle\sum_{\nu=1}^{s_1}\dfrac{\partial\,
\Phi_\alpha}{\partial\,x_\nu}
\dfrac{\partial\,\Phi_\nu}{\partial\dot{y}_k}\right)
\right)\Psi_\alpha +\displaystyle\sum_{\nu=1}^{s_1}\dfrac{\partial
L^* }{\partial\,x_\nu} \dfrac{\partial \Phi\nu}{\partial\dot{y}_k},
\end{equation}
where
$\Psi_\alpha=\left.\dfrac{\partial\,\tilde{L}}{\partial\dot{x}_\alpha}
\right|_{\dot{x}_1=\Phi_1,\ldots,\dot{x}_{s_1}=\Phi_{s_1}}.$
 The transpositional relations \eqref{211} in this
case are
\begin{equation}\label{49}
\begin{array}{rl}
\delta\dfrac{dx_\alpha}{dt}-\dfrac{d}{dt}\delta\,x_\alpha
=&\displaystyle\sum_{k=1}^{s_2}\left(\displaystyle\sum_{j=1}^{s_1}
E_j(L_\alpha )\frac{\partial{L_j}}{\partial{\dot{y_k}}}+
E_k(L_\alpha )\right)\delta y_k,\quad \alpha=1,\ldots,s_1,\vspace{0.2cm}\\
\delta\dfrac{dy_m}{dt}-\dfrac{d}{dt}\delta\,y_m=&0,\quad
m=1,\ldots,s_2.\end{array}
\end{equation}
\end{theorem}

\begin{proposition}\label{aa1} Differential equations \eqref{C48}
describe the motion of the nonholonomic systems with  { the}
constraints \,
$L_\alpha=\dot{x}_\alpha-\Phi_\alpha(\textbf{x},\textbf{y},\dot{\textbf{y}})=0$
for $\alpha=1,\ldots,s_1.$ In particular if the constraints are
given by the formula \begin{equation}\label{Vor}
\dot{x}_j=\sum_{k=1}^{s_2}a_{jk}(t,\textbf{x},\textbf{y})\dot{y}_k+a_j(t,\textbf{x}),\quad
j=1,\ldots,s_1,
\end{equation}
then systems \eqref{C48} becomes
\[\begin{array}{rl}
E_{k}(L^*)=&\displaystyle\sum_{\alpha=1}^{s_1}\left(\dfrac{d
a_{\alpha\,k} }{dt}-
\left(\dfrac{\partial{a_{\alpha\,m}}}{\partial\,{y}_k}+
\displaystyle\sum_{\nu=1}^{s_1}\dfrac{\partial\,a_{\alpha\,m}}{\partial\,x_\nu}
a_{\nu\,k}\right)\dot{y}_m\,\right) \Psi_\alpha+
\displaystyle\sum_{\nu=1}^{s_1}\dfrac{\partial L^*
}{\partial\,x_\nu} a_{\nu\,k},
\end{array}
\] which are the
classical Voronets differential equations. Consequently equations
\eqref{C48} are an extension of the Voronets differential
equations for the case when the constraints are nonlinear in the
velocities.
\end{proposition}
\begin{proposition}\label{aa2}
Differential equations \eqref{C48} describe the motion of the
constrained Lagrangian systems with the constraints
$L_\alpha=\dot{x}_\alpha-\Phi_\alpha(\textbf{y},\dot{\textbf{y}})=0$
and Lagrangian $L^*=L^*(\textbf{y},\dot{\textbf{y}}).$  Under these assumptions equations \eqref{C48} take the form
\begin{equation}\label{50}
E_{k}(L^*)=\displaystyle\sum_{\alpha=1}^{s_1}\left(\dfrac{d}{dt}
\left(\dfrac{\partial\Phi_\alpha}{\partial\dot{y}_k}\right)-
\dfrac{\partial\Phi_\alpha}{\partial\,y_k}\right)\Psi_\alpha.
\end{equation}
 In
particular if the constraints are given by the formula
\begin{equation}\label{Ch2}
\dot{x}_\alpha=\sum_{k=1}^{s_2}a_{\alpha\,k}(\textbf{y})\dot{y}_k,\quad
\alpha=1,\ldots,s_1,
\end{equation}
 then systems \eqref{50} becomes
\begin{equation}\label{Chh2}
E_{k}L^*
=\sum_{j=1}^{s_1}\sum_{r=1}^{s_2}\left(\dfrac{\partial\alpha_{jk}}{\partial{y_r}}-
\dfrac{\partial\alpha_{jr}}{\partial{y_k}}\right)\dot{y}_r\Psi_j,
\end{equation}
 for $ k=1,\ldots,s_2,$ which are the
equations which  Chaplygin published in the Proceeding of the
Society of the Friends of Natural Science in 1897 .

Consequently equations \eqref{50} are an extension of the
classical Chaplygin equations for the case when the constraints
are nonlinear.
\end{proposition}
{From}   \eqref{TT} and in view of the Implicit Function Theorem, we
can locally express the constraints (reordering  coordinates if is
necessary) as
\begin{equation}\label{qpr}
\dot{x}_\alpha=\Phi_\alpha\left(\textbf{x},\dot{x}_{M+1},\ldots,\dot{x}_N\right)
\end{equation}
for $\alpha=1,\ldots,M.$  We note that  Propositions \ref{aa1} and
\ref{aa2} are also valid for every constrained mechanical systems
with  constraints locally given by \eqref{qpr}, this follows from
Theorem 4 changing the notations, see Corollary 22.

\smallskip

The proofs of Theorem \ref{A0} and Propositions \ref{aa1} and
\ref{aa2} is given in section 8.

\smallskip

The next result is the {\it third point of view on the
transpositional relations.}
\begin{corollary}\label{CC}
For the constrained  mechanical systems the  virtual variations can
produce zero or non--zero transpositional relations.  For the
unconstrained  mechanical systems the virtual
 variations always produce zero transpositional relations.
\end{corollary}
The proof  of this corollary is given in section 9.

\smallskip

 We have the following conjecture.
\begin{conjecture} \label{CC11}
 {The existence of mechanical systems with
nonlinear  constraints in the velocity  must be sought outside of
the Newtonian model.}
\end{conjecture}
{This conjecture is supported by several facts see  section 9}.

\smallskip

The results are illustrated with precise examples.
\section{Variational Principles. Transpositional relations}
\subsection{ Hamiltonian principle }

 We introduce the following results, notations and definitions
which we will use later on (see \cite{Arnold}).

 A {\it Lagrangian  system} is a pair  $(\textsc{Q},
\tilde{L})$ consisting of a smooth manifold $ \textsc{Q},$ and a
smooth function $\tilde{L}:\mathbb{R} \times
T\textsc{Q}\longrightarrow \R,$ where $T\textsc{Q}$ is the tangent
bundle of $ \textsc{Q}.$ The point ${\bf
x}=\left(x_1,\ldots,x_N\right)\in\textsc{Q}$ denotes the {\it
position} (usually its components are called {\it generalized
coordinates}) of the system and we call each tangent vector
$\dot{{\bf x}}=\left(\dot{x}_1,\ldots,\dot{x}_N\right)\in
T_{\textbf{x}}\textsc{Q}$ the {\it velocity} (usually called {\it
generalized velocity}) of the system at the point ${\bf x}.$ A
pair $({\bf x},\dot{{\bf x}})$ is called a {\it state} of the
system. In Lagrangian mechanics it is usual to call $
\textsc{Q},$ the {\it configuration space}, the tangent bundle
$T\textsc{Q}$ is called the {\it phase space}, $ \tilde{L}$ is the
{\it Lagrange function} or {\it Lagrangian} and the dimension $N$
of $\textsc{Q}$ is the number of {\it degrees of freedom}.

Let $a_0$ and $a_1$ be two points of $ \textsc{Q}.$ The map
\[\begin{array}{rl}
\gamma:[t_0,t_1]\subset\mathbb{R}&\longrightarrow\textsc{Q},\vspace{0.2cm}\\
  t&\longmapsto\gamma(t)=\left(x_1(t),\ldots,x_N(t)\right),
          \end{array}
          \]
such that $\gamma(t_0)=a_0,\,\gamma(t_1)=a_1$ is called a {\it
path}  from $a_0$ to $a_1.$ We denote the set of all these path by
$\Omega(\textsc{Q},a_0,a_1,t_0,t_1):=\Omega$.

\smallskip

{We shall } derive  one of the most simplest and general
 variational principles the  {\it Hamiltonian principle} (see \cite{Polak}).

 The functional
$F:\Omega\longrightarrow \mathbb{R}$ defined by
\[
F(\gamma (t))=\displaystyle\int_{\gamma (t)}\tilde{L}dt=
\displaystyle\int_{t_0}^{t_1}\tilde{L}(t,\textbf{x}(t),\dot{\textbf{x}}(t))dt
\]
is called the {\it action}.

We consider the path
$\gamma(t)=\textbf{x}(t)=\left(x_1(t),\ldots,x_N(t)\right)\in\Omega.$

Let the {\it variation} of the path $\gamma(t)$ be defined as a
smooth mapping
\[\begin{array}{rl}
\gamma^*:[t_0,t_1]\times(-\tau,\tau)&\longrightarrow\textsc{Q},\vspace{0.2cm}\\
  (t,\varepsilon)&\longmapsto\gamma^*(t,\varepsilon)=\textbf{x}^*(t,\varepsilon)=
  \left(x_1(t)+\varepsilon\delta{x}_1(t),\ldots,x_N(t)+\varepsilon\delta{x}_N(t)\right),
          \end{array}
          \]
          satisfying
\[
\textbf{x}^{*}(t_0,\varepsilon)=a_0,\quad\textbf{x}^{*}(t_1,\varepsilon)=a_1,\quad\textbf{x}^{*}(t,0)=\textbf{x}(t).
\]
By definition we have
\[\delta{\textbf{x}}(t)=\left.\dfrac{\partial\textbf{x}^{*}(t,\varepsilon)}{\partial\varepsilon}\right|_{\varepsilon=0}.\]
This function  is called the {\it virtual displacement} or {\it
virtual variation} corresponding to the variation of $\gamma(t)$
and it is a function of time, all  its components are functions of
$t$ of class $C^2(t_0,t_1)$ and vanish at $t_0$ and $t_1$
i.e. $\delta\textbf{x}(t_0)=\delta\textbf{x}(t_1)=0.$

\smallskip

A {\it varied path} is a path which can be obtained as a variation path.

\smallskip

 The {\it first variation} of the functional $F$ at $\gamma(t)$ is
\[\delta{F}:=\left.\dfrac{\partial
F\left(\textbf{x}^{*}(t,\varepsilon)\right)}{\partial\varepsilon}\right|_{\varepsilon=0},\]
and it is called the {\it differential } of the functional $F$ (see
\cite{Arnold}). The path $\gamma (t)\in\Omega$ is called the {\it
critical point} of $F$ if $\delta F(\gamma (t))=0.$

 Let $\mathbb{L}$ be
the space of all smooth functions $g:\mathbb{R} \times
T\textsc{Q}\longrightarrow \R.$ The operator
 \[\begin{array}{rl}
 E_\nu:\mathbb{L}&\longrightarrow \R,\\
g&\longmapsto E_\nu g={\dfrac{d}{dt}\dfrac{\partial
g}{\partial\dot{x}_\nu}-\dfrac{\partial g}{\partial{x}_\nu}},\quad
\mbox{for}\quad \nu=1,\ldots,N,
\end{array}
\]
  is known as the {\it Lagrangian
derivative.}

 It is easy to show the following property of the
Lagrangian derivative
\begin{equation}\label{An1}
E_\nu\dfrac{df}{dt}=0, \end{equation}
 for arbitrary smooth function $f=f(t,\textbf{x}).$
 We observe that in view of \eqref{An1} we obtain
that the Lagrangian derivative is unchanged if we replace the
function $g$ by $g+\dfrac{df}{dt},$ for any function
$f=f(t,\textbf{x}).$ This reflects the {\it gauge invariance.} We
shall say that the functions
$g=g\left(t,\textbf{x},\dot{\textbf{x}}\right)$
 and $\hat{g}=\hat{g}\left(t,\textbf{x},\dot{\textbf{x}}\right)$ are
 {\it equivalently} if
 $
 g-\hat{g}=\dfrac{df(t,\textbf{x})}{dt},
 $
 and we shall write $g\simeq \hat{g}.$

\begin{proposition} \label{LL}
The differential of the action can be calculated as follows
\begin{equation}\label{ttt}
\delta{F}=-\displaystyle\int_{t_0}^{t_1}\displaystyle\sum_{k=1}^N\left(
E_k\tilde{L}\delta{x}_k- \dfrac{\partial
\tilde{L}}{\partial\dot{{x_k}}}\left(\delta\dfrac{d{x_k}}{dt}
-\dfrac{d}{dt}\delta{x_k}\right)\right)dt,
\end{equation}
 where
${\bf{x}}={\bf{x}}(t),\,\dot{\bf{x}}=\dfrac{d{\bf{x}}}{dt},$ and
$\tilde{L}=\tilde{L}\left(t,{\bf{x}},\dfrac{d{{\bf{x}}}}{dt}\right).$
\end{proposition}

\begin{proof} We have that
\[
\begin{array}{rl}
\delta{F}=&\left.\dfrac{\partial
F\left(\textbf{x}^{*}(t,\varepsilon)\right)}{\partial\varepsilon}\right|_{\varepsilon=0}
\vspace{0.2cm}\\
=&\displaystyle\int_{t_0}^{t_1}\left.\dfrac{\partial}{\partial\varepsilon}\right|_{\varepsilon=0}
L\left(t,\textbf{x}^{*}(t,
\varepsilon),\dfrac{d}{dt}\left(\textbf{x}^{*}(t,\varepsilon)\right)\right)
\,dt=
\displaystyle\int_{t_0}^{t_1}\displaystyle\sum_{k=1}^N\left(\dfrac{\partial
L}{\partial{x_k}}\delta{x_k}+\dfrac{\partial
L}{\partial\dot{{x_k}}}\delta{\dot{x_k}}\right)dt\vspace{0.2cm}\\
=&
\displaystyle\int_{t_0}^{t_1}\displaystyle\sum_{k=1}^N\left(\dfrac{\partial
L}{\partial{{x_k}}}\delta{x_k}+\dfrac{\partial
L}{\partial\dot{{x_k}}}\dfrac{d}{dt}\delta x_k+ \dfrac{\partial
L}{\partial\dot{{x_k}}}\left(\delta\dfrac{d{x_k}}{dt}-
\dfrac{d}{dt}\delta{x_k}\right)\right)dt\vspace{0.3cm}\\
=&\left.\displaystyle\sum_{k=1}^N\dfrac{\partial
L}{\partial\dot{{x_k}}}\delta{x_k}\right|_{t=t_0}^{t=t_1}+
\displaystyle\int_{t_0}^{t_1}\displaystyle\sum_{k=1}^N\left(\left(\dfrac{\partial
L}{\partial{{x_k}}}-\dfrac{d}{dt}\dfrac{\partial
L}{\partial\dot{{x_k}}}\right)\delta{{x_k}}+
\dfrac{\partial
L}{\partial\dot{{x_k}}}\left(\delta\dfrac{d{x_k}}{dt}
-\dfrac{d}{dt}\delta{x_k}\right)\right)dt.
\end{array}
\]
Hence, by considering that the virtual variation vanishes at the
points $t=t_0$ and $t=t_1$ we obtain the proof of the proposition.
\end{proof}
\begin{corollary}\label{ll}
The differential of the action for a Lagrangian system $\left(\textsc{Q},\,\tilde{L}\right)$ can be
calculated as follows
\[
\delta{F}=-\displaystyle\int_{t_0}^{t_1}\displaystyle\sum_{k=1}^N
E_k\tilde{L}\left(t,{\bf{x}},\dfrac{d{\bf{x}}}{dt}\right)\,\delta{x}_k\,dt.
\]
\end{corollary}
\begin{proof}
Indeed,  for the Lagrangian system the transpositional relation is
equal to zero (see for instance \cite{Lurie} page 29), i.e.
\begin{equation}\label{1313}
\delta\dfrac{d\textbf{x}}{dt}-\dfrac{d}{dt}\delta\textbf{x}=0.
\end{equation}
Thus, from Proposition \ref{LL}, it follows the proof of the corollary.
\end{proof}
  The path $\gamma (t)\in\Omega$ is called a {\it motion} of the
Lagrangian systems $\left(\textsc{Q},\,\tilde{L}\right)$ if
$\gamma (t)$ is a { critical point} of the action $F,$ i.e.
\[
\delta{F}\left(\gamma
(t)\right)=0\Longleftrightarrow\displaystyle\int_{t_0}^{t_1}\delta{\tilde{L}}\,dt=0.
\]

 This definition is known as the {\it Hamiltonian variational principle}
 or {\it Hamiltonian variational
principle of least action} or simple {\it Hamiltonian principle}.

\smallskip

Now we need the {\it Lagrange lemma} or {\it fundamental lemma of
calculus of variations} (see for instance \cite{Alekciev})
\begin{lemma}\label{Lagrange}
Let $f$ be a continuous function of the interval $[t_0,\,t_1]$
 satisfying the equation
\[\displaystyle\int_{t_0}^{t_1}f(t)\zeta{(t)}dt=0,\] for arbitrary continuous function
$\zeta(t)$  such that $\zeta (t_0)=\zeta (t_1)=0.$ Then $f(t)\equiv 0.$
\end{lemma}
\begin{corollary}
 The Hamiltonian
principle for Lagrangian systems is equivalent to the Lagrangian
equations
\begin{equation}\label{Lag}
E_\nu\tilde{L}=\displaystyle\dfrac{d}{dt}\left(\dfrac{\partial
\tilde{L}}{\partial\dot{{x}_\nu}}\right)-\dfrac{\partial\tilde{L}
}{\partial{{x}_\nu}}=0,
\end{equation}
for $\nu=1,\ldots,N.$
\end{corollary}
\begin{proof}
Clearly, if \eqref{Lag} holds, by Corollary \ref{ll},  $\delta{F}=0.$ The reciprocal
result follows from Lemma \ref{Lagrange}.
\end{proof}
 {F}rom the formal point of view, the Hamiltonian principle in the
form \eqref{LLag} is equivalent to the problem of {\it variational
calculus} \cite{Gelfand,Polak}. However, despite the superficial
similarity, they differ essentially. Namely, in mechanics the symbol
$\delta$ stands for the {its virtual variation}, i.e., it is not an
arbitrary  variation but a displacement compatible with the
constraints imposed on the systems. Thus only in the case of the
holonomic systems, for which the number of degrees of freedom is
equal to the number of generalized coordinates, the virtual
variations are arbitrary and the Hamiltonian principle \eqref{LLag}
is completely equivalent to the corresponding problem of the
variational calculus. An important difference arises for the systems
with nonholonomic constraints, when the variations of the
generalized coordinates are connected by the additional relations
usually called Chetaev conditions which we give later on.

\subsection{ D'Alembert--Lagrange principle}

 Let $L_j:\mathbb{R}\times{T\textsc{Q}}\longrightarrow\mathbb{R}$ be smooth functions for $j=1,\ldots,M.$ The equations
\[
L_j=L_j\left(t,\bf{x},\dot{\bf{x}}\right)=0,\quad \mbox{for}\quad
j=1,\ldots,M< N, \]
with\,$\mbox{rank}\left(\dfrac{\partial(L_1,\ldots,L_M)}
{\partial(\dot{x}_1,\ldots,\dot{x}_N)} \right)=M$ in all the points
of $ \mathbb{R}\times T\textsc{Q},$ except perhaps in a zero
Lebesgue measure set, define $M$ {\it independent constraints} for
the Lagrangian systems $ (\textsc{Q},\tilde{L}).$

\smallskip

Let $\MM^*$ be the submanifold of $\mathbb{R}\times T\textsc{Q}$
defined by the equations \eqref{01111}, i.e.
\[
\MM^*=\{\left({t,\bf x},\,\dot{{\bf x}}\right)\in \mathbb{R}\times
T\textsc{Q}: L_j({t,\bf x},\dot{\bf{ x}})=0,\quad \mbox{for}\quad
j=1,\ldots,M\}.
\]
A {\it constrained Lagrangian system} is a triplet
$(\textsc{Q},\tilde{L},\MM^*).$  The number of degree of freedom
is $\kappa=dim{\textsc{Q}}-M=N-M.$

The constraint is  called  {\it integrable } if it
 can be written in the form
$L_j=\dfrac{d}{dt}\left(G_j(t,\textbf{x})\right)=0,$ for
a convenient function $G_j.$ Otherwise the constraint is called {\it
nonintegrable}. According to Hertz \cite{Hertz} the nonintegrable
constraints are also called {\it nonholonomic.}

The Lagrangian systems with nonintegrable constraints are usually
called (also following to Hertz) the {\it nonholonomic mechanical
systems,} or {\it nonholonomic constrained mechanical systems,} and
with integrable constraints are called the {\it holonomic
constrained mechanical systems} or {\it holonomic constrained
Lagrangian systems}. The systems free of constraints are called
{\it Lagrangian systems or holonomic systems.}

\smallskip

 Sometimes it is also useful to distinguish between constraints
that are dependent on or independent of time.  Those that are
independent of time are called  {\it scleronomic,}  and those that
depend on time are called  {\it rheonomic}. This therminology can
also be applied to the mechanical systems themselves. Thus we say that
the constrained Lagrangian systems is scleronomic (reonomic) if
the constraints and Lagrangian are time independent (dependent).

 The constraints
\begin{equation}
\label{1010}
L_k=\displaystyle\sum_{j=1}^Na_{kj}\dot{x}_j+a_k=0,\quad
\mbox{for}\quad k=1,\ldots,M,
\end{equation}
 where
$a_{kj}=a_{kj}(t,\textbf{x}),\,\,a_k=a_k(t,\textbf{x}),$ are called
{\it linear constraints with respect to the velocity}. For
simplicity we shall call {\it linear constraints.}

\smallskip

We observe that \eqref{1010} admits an equivalent representation
as a Pfaffian equations (for more details see \cite{Pars})
\[\omega_k:=\displaystyle\sum_{j=1}^Na_{kj}dx_j+a_k\,dt=0.\]
We shall consider only two classes of systems of equations, the equations
of constraints linear with respect to the velocity
$(\dot{x}_1,\ldots,\dot{x}_N)$, or linear with respect to the
differential $(dx_1,\ldots,dx_N,dt).$  In order to study the integrability
or nonintegrability problem of the constraints the last representation,
 a Pfaffian system is the  more useful.
 This is related with the fact that for the given 1-forms we have
 the Frobenius theorem which provides the necessary and
 sufficient conditions under which the 1-forms are
closed and consequently the given set of constraints is
integrable.

 The constrains  $L_j({t,\bf
x},\dot{\bf{ x}})=0$  are called {\it perfect constraints} or {\it
ideal} if they satisfy the {\it Chetaev conditions} (see \cite{Chetaev1})
\begin{equation}
\label{103}
 \displaystyle\sum_{k=1}^N\dfrac{\partial
L_\alpha}{\partial\dot{x}_k}\,\delta x_k=0,
\end{equation}
for $\alpha=1,\ldots,M.$

\smallskip

{\it In what follows, we shall consider only perfect constraints.}

If the constraints admit the representation \eqref{qpr} then the
Chetaev conditions takes the form
\[\delta x_\alpha=\displaystyle\sum_{k=M+1}^N\dfrac{\partial
\Phi_\alpha}{\partial\dot{x}_k}\delta x_k.\] The virtual variations
of the variables $x_\alpha$ for $\alpha=1,\ldots,M$ are called {\it
dependent variations} and for the variable $x_\beta$ for
$\beta=M+1,\ldots,N$ are called {\it independent variations}.

\smallskip

We say that the path $\gamma (t)=\textbf{x}(t)$ is {\it admissible}
 with the perfect constraint if
$L_j({t,\bf x}(t),\dot{\bf{ x}}(t))\,\,=0.$

\smallskip

The admissible path is called the {\it motion} of the constrained
Lagrangian systems $(\textsc{Q},\tilde{L},\MM^*)$ if for all
$t\in[t_0,t_1]$
\[\displaystyle\sum_{\nu=1}^N
E_\nu\tilde{L}\left({t,\bf x}(t),\dot{\bf{ x}}(t)\right)\,\delta{x}_\nu(t)=0,
\]
 for all virtual displacement
$\delta{\textbf{x}}(t)$ of the path $\gamma (t).$  This definition
is known as {\it d'Alembert--Lagrange principle}.

It is well known the following result (see for instance
\cite{Arnold,Bloch,Griffiths, NF}).
\begin{proposition}
The d'Alembert--Lagrange principle for constrained Lagrangian
systems is equivalent to the {\it Lagrangian differential
equations with multipliers}
\begin{equation}\label{35}
 \begin{array}{rl}
 E_j\tilde{L}=&{\dfrac{d}{dt}\dfrac{\partial
\tilde{L}}{\partial{\dot{x}_j}}-\dfrac{\partial
\tilde{L}}{\partial{{x}_j}}}=\displaystyle\sum_{\alpha=1}^M{\mu_\alpha\dfrac{\partial
L_\alpha}{\partial{\dot{x_j}}}}, \quad \mbox{for}\quad
j=1,\ldots,N,\vspace{0.2cm}\\
L_j({t,\bf x},\dot{\bf{ x}})=&0,\quad \mbox{for}\quad
j=1,\ldots,M,
\end{array}
\end{equation}
where  $\mu_\alpha$ for $\alpha=1,\ldots,M$ are the Lagrangian
multipliers.
\end{proposition}

\subsection{The varied path}

The varied path produced in Hamiltonian's principle is not in
general an admissible path if the perfect constraints are nonholonomic,
i.e. the mechanical systems cannot travel along the
varied path without violating the constraints. We prove the
following result, which shall play an important role in the all
assertions below.
\begin{proposition}\label{QP}
If the varied path is an admissible path then, the following
relations hold
\begin{equation}\label{Nat1}
\displaystyle\sum_{k=1}^N\dfrac{\partial
L_\alpha}{\partial\dot{x}_k}\left(\delta \dfrac{d
x_k}{dt}-\dfrac{d}{dt}\delta
x_k\right)=\displaystyle\sum_{k=1}^NE_kL_\alpha\,\delta{x_k},
\end{equation}
for $\alpha=1,\ldots,M.$
\end{proposition}

\begin{proof}
Indeed, the original path $\gamma(t)=\textbf{x}(t)$ by definition
satisfies the Chetaev conditions, and constraints, i.e.
$L_j\left(t,\textbf{x}(t),\dot{\textbf{x}}(t)\right)=0.$ If we
suppose that the variation path
$\gamma^*(t)=\textbf{x}(t)+\varepsilon\delta{\textbf{x}}(t),$ also
satisfies the constraints i.e.
\[
L_j\left(t,\textbf{x}+\varepsilon\delta\textbf{x},
\dot{\textbf{x}}+\varepsilon\delta\dot{\textbf{x}}\right)=
L_j\left(t,\textbf{x}(t),\dot{\textbf{x}}(t)\right)+
\varepsilon\delta\,L_\alpha\left(t,\textbf{x}(t),
\dot{\textbf{x}}(t)\right)+\ldots=0.
\]
Thus restricting only to the
 terms of first order with respect
 to $\varepsilon$ and by the Chetaev conditions we have (for simplicity we omitted the argument)
\begin{equation}\label{Nat}
\begin{array}{rl}
0=&\delta\,L_\alpha=\displaystyle\sum_{k=1}^N\left(\dfrac{\partial
L_\alpha}
{\partial\,{x}_k}\delta{{x}_k}+\dfrac{\partial
L_\alpha}
{\partial\dot{{x}_k}}\delta{\dot{{x}_k}}\right),\vspace{0.2cm}\\
0=&\displaystyle\sum_{k=1}^N\dfrac{\partial
L_\alpha}{\partial\dot{x}_k}\,\delta
x_k,
\end{array}
\end{equation}
for $ \alpha=1,\ldots,M.$ The Chetaev conditions are
satisfied at each instant, so
\[
\dfrac{d}{dt}\left(\displaystyle\sum_{k=1}^N\dfrac{\partial
L_\alpha}{\partial\dot{x}_k}\,\delta
x_k\right)=\displaystyle\sum_{k=1}^N\dfrac{d}{dt}\left(\dfrac{\partial
L_\alpha}{\partial\dot{x}_k}\right)\,\delta
x_k+\displaystyle\sum_{k=1}^N\dfrac{\partial
L_\alpha}{\partial\dot{x}_k}\dfrac{d}{dt}\delta x_k=0.
\]
 Subtracting these relations from \eqref{Nat} we
obtain \eqref{Nat1}. Consequently if the varied path is an
admissible path, then relations \eqref{Nat1} must hold.
\end{proof}

\smallskip

From \eqref{Nat1} and \eqref{23} it follows that the elements of the matrix $A$ satisfy
\begin{equation}\label{ff1}
\displaystyle\sum_{m=1}^N\delta{x_m}\left(E_mL_\alpha-
\displaystyle\sum_{k=1}^NA_{k\,m}\dfrac{\partial
L_\alpha}{\partial\dot{x}_k}\right)=\displaystyle\sum_{m=1}^N\delta{x_m}D_mL_\alpha= 0,\quad\mbox{for}\quad
\alpha=1,\ldots,M.
\end{equation}
  This property will be used below.

\smallskip

\begin{corollary}\label{main11}
For the holonomic constrained Lagrangian systems the relations
\eqref{Nat1} hold if and only if
\begin{equation}\label{Val}
\displaystyle\sum_{k=1}^N\dfrac{\partial
L_\alpha}{\partial\dot{x}_k}\left(\delta \dfrac{d
x_k}{dt}-\dfrac{d}{dt}\delta x_k\right)=0,\quad\mbox{for}\quad
\alpha=1,\ldots,M.
\end{equation}
\end{corollary}
\begin{proof}
Indeed, for holonomic constrained Lagrangian systems the
constraints are integrable, consequently in view of \eqref{An1}
we have $E_kL_\alpha=0$ for $k=1,\ldots,N$ and $\alpha=1,\ldots,M.$  Thus,
from \eqref{Nat1}, we obtain \eqref{Val}.
\end{proof}
Clearly the equalities  \eqref{Val} are satisfied  if \eqref{1313}
holds. We observe that in general for holonomic constrained
Lagrangian systems relation  \eqref{1313} cannot hold (see example
2).

\subsection{Transpositional relations}

 As we observe in the previous subsection for nonholonomic constrained Lagrangian
 systems the curves, obtained doing a  virtual
variation in the motion of the systems, in general are not
kinematical possible trajectories when  \eqref{1313} {is not
fulfilled}. This leads to the conclusion that the Hamiltonian
principle cannot be applied to nonholonomic systems, as it is
usually employed for holonomic systems. The essence of the problem
of the applicability of this principle for nonholonomic systems
remains unclarified (see \cite{NF}). In order to clarify this
situation, it is sufficient to note that the question of the
applicability of the principle of stationary action to nonholonomic
systems is intimately related to the question of transpositional
relation.

 The key point is that the Hamiltonian
principle assumes that the operation of differentiation with
respect to the time $\dfrac{d }{dt}$ and the virtual variation
$\delta $ commute  in  all the generalized coordinate systems.

For the holonomic constrained Lagrangian systems relations
\eqref{1313} cannot hold (see Corollary \ref{main11}). For
 a nonholonomic systems the form of the
Hamiltonian principle will depend on the  point of view
adopted with respect to the transpositional relations.

 What are then the correct transpositional relations? Until
now, does not exist a common point of view concerning to
 the commutativity of the operation of differentiation with respect to the time
and the virtual variation when there are nonintegrable
constraints. Two points of view have been maintained. According to
one (supported, for example, by Volterra,
Hamel, H{\"o}lder,  Lurie,\,Pars,\ldots), the operations $\dfrac{d}{dt}$ and
$\delta$ commute for all the generalized coordinates, independently if
 the systems are holonomic or nonholonomic, i.e.
\[\delta \dfrac{d
x_k}{dt}-\dfrac{d}{dt}\delta x_k=0,\quad \mbox{for}\quad
k=1,\ldots,N.\]  According to the other point of view (supported by
Suslov, Voronets, Levi-Civita, Amaldi,\ldots) the operations
$\dfrac{d}{dt}$ and $\delta$ commute always for holonomic systems,
and for nonholonomic systems with the constraints
\[\dot{x}_\alpha=\displaystyle\sum_{j=M+1}^Na_{\alpha
j}(t,\textbf{x})\dot{x}_j+a_\alpha(t,\textbf{x}),\quad\mbox{for}\quad
\alpha=1,\ldots,M.\]
 the transpositional relations are equal to zero only for the
generalized coordinates $x_{M+1},\ldots,x_N,$ ( for which their
virtual variations are independent). For the remaining coordinates
$x_1,\ldots,x_M,$ (for which their virtual variations are
dependent), the transpositional relations must be derived on the
basis of the equations of the nonholonomic constraints, and cannot
 be identically\,zero, i.e.
\[\begin{array}{rl}
\delta \dfrac{d x_k}{dt}-\dfrac{d}{dt}\delta x_k=&0,\quad
\mbox{for}\quad k=M+1,\ldots,N\\
\delta \dfrac{d x_k}{dt}-\dfrac{d}{dt}\delta x_k\ne&0,\quad
\mbox{for}\quad k=1,\ldots,M.\end{array}\]
  The second point of view acquired general
acceptance and the first point of view was considered erroneous
(for more details see \cite{NF}).
 The meaning of the transpositional relations \eqref{13} can be found in
 \cite{Kirgetov,Lurie,Marushin,NF}.

\smallskip

In the results given in the following section play a key role the
equalities \eqref{Nat1}. From these equalities  and from the
examples it will be  possible to observe that the second point of
view is correct only for the so called Voronets--Chaplygin {systems,
and in general  for  locally  nonholonomic systems. There exist many
examples for which the independent virtual variations generated
non--zero transpositional relations. Thus we propose a third point
of view on the transpositional relations: the virtual variations can
generate the transpositional relations given by the formula
\eqref{23} where the elements of the matrix $A$ satisfies the
conditions (see formula \eqref{ff1})

\begin{equation}\label{ValF}
D_\nu L_\alpha=E_\nu L_\alpha-
\displaystyle\sum_{k=1}^NA_{k\,\nu}\dfrac{\partial
L_\alpha}{\partial\dot{x}_k}   = 0,\quad \mbox{for}\quad \nu
=1,\ldots,M,\quad \alpha =1,\ldots,M.\end{equation} we observe that
here the $L_\alpha=0$ are constraints which in general are nonlinear
in the velocity.

\smallskip

\subsection{Hamiltonian--Suslov principle}
 After the  introduction of  the nonholonomic mechanics by Hertz,
  it appeared the question  of extending to the nonholonomic
   mechanics the results of  the holonomic mechanics. Hertz
\cite{Hertz}  was the first in studying the problem of
applying  the Hamiltonian principle to  systems with
nonintegrable constraints. In \cite{Hertz} Hertz wrote:
``Application of Hamilton's principle to any material systems does
not exclude that between selected coordinates of the systems rigid
constraints exist, but it still requires that these relations
could be expressed by integrable constraints. The appearance of
nonintegrable constraints  is unacceptable. In this case  the
Hamilton's principle is not valid."
  Appell \cite{Appell} in correspondence with
Hertz's ideas affirmed that it is not possible to apply the
Hamiltonian principle for systems with nonintegrable constraints

Suslov  \cite{Suslov} claimed that ''Hamilton's principle is not
applied to systems with  nonintegrable constraints, as derived based
on this equation are different from the corresponding equations of
Newtonian mechanics".

The applications of the most general differential  principle, i.e. the
d'Alembert--Lagrange and their equivalent Gauss and Appel
principle, is complicated due to the presence of the terms containing
the second order derivative. On the other hand the most general
variational integral principle of Hamilton is not valid for
nonholonomic constrained Lagrangian systems. The generalization of
the Hamiltonian principle for nonholonomic
 mechanical systems  was deduced  by Voronets and
 Suslov (see for instance \cite{Suslov,Voronets}).  As we can observe later on from this
principle follows the importance of the transpositional relations
to determine the correct equations of motion for nonholonomic
constrained Lagrangian systems.

\begin{proposition}
The d'Alembert--Lagrangian principle for the contrained Lagrangian
systems $\displaystyle\sum_{k=1}^N\delta{x}_kE_k\tilde{L}=0$
 is equivalent to the Hamilton--Suslov principle \eqref{Su11}
where we  assume that $\delta{x_\nu(t)},\quad \nu=1,\ldots, N,$
are arbitrary smooth functions defined in the interior of the
interval $[t_0,\,t_1]$ and  vanishing at its endpoints, i.e.,
$\delta{x_\nu}({t_0})=\delta{x_\nu}({t_1})=0.$
\end{proposition}
\begin{proof}
{F}rom the
d'Alembert--Lagrangian principle we obtain the identity
\[\begin{array}{rl}
0=&-\displaystyle\sum_{k=1}^N\delta{x}_kE_k\tilde{L}=
\displaystyle\sum_{k=1}^N\delta{x}_k\dfrac{\partial
\tilde{L}}{\partial
x_k}-\displaystyle\sum_{k=1}^N\delta{x}_k\dfrac{d}{dt}\dfrac{\partial
\tilde{L}}{\partial\dot{x}_k}\vspace{0.2cm}\\
=&\displaystyle\sum_{k=1}^N\left(\delta{x}_k\dfrac{\partial
\tilde{L}}{\partial x_k}+\delta{\dot{x}}_k\dfrac{\partial
\tilde{L}}{\partial\dot{x}_k}\right)-
\displaystyle\sum_{k=1}^N\left(\left(\delta\dfrac{dx_k}{dt}-
\dfrac{d}{dt}\delta{x_k}\right)\dfrac{\partial \tilde{L}}{\partial
\dot{x}_k}-\dfrac{d}{dt}\left(\dfrac{\partial \tilde{L}}{\partial
\dot{x}_k}\delta{x_k}\right)\right)\vspace{0.2cm}\\
=&\delta\tilde{L}-\displaystyle\sum_{k=1}^N\left(\left(\delta\dfrac{dx_k}{dt}-
\dfrac{d}{dt}\delta{x_k}\right)\dfrac{\partial \tilde{L}}{\partial
\dot{x}_k}-\dfrac{d}{dt}\left(\dfrac{\partial \tilde{L}}{\partial
\dot{x}_k}\delta{x_k}\right)\right),
\end{array}
\]
where $\delta \tilde{L}$ is a variation of the Lagrangian $\tilde{L}$.
 After the integration and
assuming that $\delta x_k(t_0)= 0,\,\delta x_k(t_1)=0$ we easily
obtain \eqref{Su11}, which represent the most general formulation of the Hamiltonian
principle ({\it{Hamilton--Suslov principle}}) suitable for
constrained and unconstrained Lagrangian systems.
\end{proof}
 Suslov determine the
transpositional relations only for the case when the constraints are
of Voronets type, i.e. given by the formula \eqref{Vor}.
  Assume that
\[\delta\dfrac{dy_k}{dt}-
\dfrac{d}{dt}\delta{y_k}=0,\quad\mbox{for}\quad k=M+1,\ldots,N,\]
Voronets and Suslov  deduced that
\[\delta\dfrac{dx_k}{dt}-
\dfrac{d}{dt}\delta{x_k}=\displaystyle\sum_{k=1}^NB_{kr}\delta{y_r}-\delta{a_k}\]
for convenient functions
$B_{kr}=B_{kr}\left(t,\textbf{x},\textbf{y},\dot{\textbf{x}},\dot{\textbf{y}}\right),$
for $r=M+1,\ldots,N$ and $k=1,\ldots,M.$

Thus we obtain
\[
\displaystyle\int_{t_0}^{t_1}\left(\delta\,\tilde{L}-
\displaystyle\sum_{k=1}^N\dfrac{\partial{\tilde{L}}}{\partial\dot{x}_j}
\left( \displaystyle\sum_{k=1}^NB_{kr}\delta{y_r}-\delta{a_k}
\right)\right)dt=0,
\]
This is the Hamiltonian principle for nonholonomic systems in the
Suslov form  (see for instance \cite{Suslov}). We observe that the
same result was deduced by Voronets in \cite{Voronets}.

It is important to observe that Suslov and Voronets require a priori
that the independent virtual variations produce the zero
transpositional relations. At the sometimes these authors consider
only linear  constraints with respect to the velocity of the type
\eqref{Vor}.

\subsection{Modification of the vakonomic mechanics (MVM)} As we
observe in the introduction, the main objective of this paper is
to construct the variational equations of motion  describing the
behavior of the constrained Lagrangian systems in which the
equalities \eqref{Nat1} take place in the most general possible
way.  We shall show that the d'Alembert--Lagrange
 principle is not the only way to deduce the equations of
 motion for the constrained Lagrangian systems. Instead of it we can
 apply  the generalization of the Hamiltonian  principle, whereby
 the motions of such systems are extremals of the {\it variational
 Lagrange problem} (see for instance \cite{Gelfand}), i.e. the problem of determining the
 critical points of the action in the class of curves with fixed
 endpoints and satisfying the constraints. The solution of this
 problem as we shall see will give the differential equations of second order
 which coincide with the well--known classical equations of the
 mechanics except perhaps in a zero Lebesgue measure set.

\smallskip

From the previous section we deduce that in order to generalize
the Hamiltonian principle to nonholonomic systems we must take
into account the following relations
\[
\begin{array}{rl}
&\mbox{(A)}\qquad\delta
L_\alpha=\displaystyle\sum_{j=1}^N\left(\dfrac{\partial
L_\alpha}{\partial x_j}\delta x_j+\dfrac{\partial
L_\alpha}{\partial \dot{x}_j}\delta
\dot{x}_j\right)=0\quad\mbox{for}\quad
\alpha=1,\ldots,M,\vspace{0.2cm}\\
&\mbox{(B)}\qquad\displaystyle\sum_{j=1}^N\dfrac{\partial
L_\alpha}{\partial \dot{x}_j}\delta{x_j}=0\quad\mbox{for}\quad
\alpha=1,\ldots,M,\vspace{0.2cm}\\
&\mbox{(C)}\qquad\delta\dfrac{dx_j}{dt}-
\dfrac{d}{dt}\delta{x_j}=0\quad\mbox{for}\quad j=1,\ldots,N,
\end{array}
\]
where $L_\alpha=0$ for $\alpha=1,\ldots,M$ are the constraints.

\smallskip

A lot of authors consider that (C)
 is always fulfilled (see for instance
\cite{Lurie,Pars}), together with the conditions (A) and (B).
However these conditions are incompatible in the case of the
nonintegrable constrains.
 We observe that these authors deduced that the Hamiltonian principle is not applicable to the
nonholonomic systems.

 To obtain a generalization of the Hamiltonian principle for the
 nonholonomic mechanical systems, some of these three conditions must be excluded.

  In particular for the H{\"o}lder principle conditions (A) is excluded and keep (B) and (C) (see \cite{Holder}).
   For the Hamiltonian--Suslov principle condition (A)
and (B) {hold, and}  (C) only holds  for the independent variations.

\smallskip

In this paper we extend  the Hamiltonian principle by supposing that
conditions (A) and (B) hold and (C) does not hold . Instead of (C)
we consider that \eqref{23} holds where elements of matrix $A$
satisfy the relations \eqref{ValF}.

\section{Solution of the inverse problem of the constrained Lagrangian systems}

We shall determine  the equations of motion of the  constrained
Lagrangian systems using the Hamiltonian principle with non zero
transpositional relations, whereby the motions of the systems are
{\it extremals} of the {\it variational Lagrange's problem} (see for
instance \cite{Gelfand}), i.e. are the critical points of the action
functional
\[\displaystyle\int_{t_0}^{t_1}L_0\left(t,\textbf{x},\dot{\textbf{x}}\right)\,dt,\]
in the class of path with fixed endpoints satisfying the independent
constraints \[
L_j\left(t,\textbf{x},\dot{\textbf{x}}\right)=0,\quad\mbox{for}\quad
j=1,\ldots,M.
\]
In the classical solution of the Lagrange problem usually we apply
the {\it Lagrange multipliers method} which consists in the
following. We introduce the additional coordinates
$\Lambda=\left(\lambda_1,\ldots,\lambda_M\right),$ and Lagrangian
 $
\widehat{L}:\mathbb{R}\times{T\textsc{Q}}\times\mathbb{R}^M
\longrightarrow\mathbb{R}
 $
 given by
\[\widehat{ L}\left(t,\textbf{x},\dot{\textbf{x}},\Lambda
\right)=L_0\left(t,\textbf{x},\dot{\textbf{x}}\right)-
\displaystyle\sum_{j=1}^M\lambda_j\,L_j\left(t,\textbf{x},
\dot{\textbf{x}}\right),
\]
Under this choice we reduce the Lagrange problem to a variational
problem without constraints, i.e. we must determine the extremal
of the action functional
$
\displaystyle\int_{t_0}^{t_1}\widehat{L}\,dt.
$
We shall study a  slight {\it modification of the Lagrangian
multipliers  method}.  We introduce the additional coordinates
$\Lambda=\left(\lambda_1,\ldots,\lambda_M\right),$ and the
Lagrangian on $\mathbb{R}\times{T\textsc{Q}}\times\mathbb{R}^M$
given by the formula \eqref{21},
 where we assume that
$\lambda^0_j$ are arbitrary constants,  and
 $L_{j}$  are arbitrary functions for $j=M+1,\ldots,N.$

\smallskip

 Now we determine the
critical points of the action functional $
\displaystyle\int_{t_0}^{t_1}L\left(t,\textbf{x},\dot{\textbf{x}},\Lambda\right)dt,
$ i.e. we determine the path $\gamma(t)$  such that
$\displaystyle\int_{t_0}^{t_1}\delta\left(L\left(t,\textbf{x},\dot{\textbf{x}},\Lambda\right)\right)dt=0$
 under the additional condition that the
transpositional  relations are  given by the formula \eqref{23}.

\smallskip

{The solution} {of the inverse problem stated in section 2 is the
following}. Differential equations obtained from \eqref{023} are
given by the formula \eqref{24} (see Theorem \ref{A}). We choose
the arbitrary functions $L_j$  in such a away that the matrix
$W_1$ and $W_2$ given in Theorems \ref{A1} and \ref{A2} are
nonsingular, except perhaps in a zero Lebesgue measure set. The constants
$\lambda^0_j$ for $j=M+1,\ldots,N$ are arbitrary  in Theorem
\ref{A1}, and $\lambda^0_j$ for $j=1,\ldots,N-1$ are arbitrary and
$\lambda^0_N=0$ in Theorem \ref{A2}. The matrix $A$ is determined
from the equalities \eqref{a1} and \eqref{a2} of Theorems \ref{A1}
and \ref{A2} respectively.

\smallskip
\begin{remark}
 It is interesting to observe that from the solutions of the inverse problem,
 the constants  $\lambda^0_j$
for $j=M+1,\ldots,N$ are arbitrary except in Theorem \ref{A2} in
which $\lambda^0_N=0.$
  Clearly, if
$L_j\left(t,\textbf{x},
\dot{\textbf{x}}\right)=\dfrac{d}{dt}f_j(t,\textbf{x})$ for
$j=M+1,\ldots,N,$ then the $L\simeq \widehat{L}.$ Using the
arbitrariness of the constants $\lambda^0_j$ we can always take
that $\lambda^0_k=0$ if $L_k\left(t,\textbf{x},
\dot{\textbf{x}}\right)\ne\dfrac{d}{dt}f_k(t,\textbf{x}).$
Consequently we can always suppose that  $L\simeq \widehat{L}.$
Thus the only difference between the classical and the modified
Lagrangian multipliers method consists only on the transpositional
relations: for the classical method the virtual variations produce
zero transpositional relations (i.e. the matrix $A$ is the zero
matrix) and for the modified method in general it is determined by
the formulae \eqref{23} and \eqref{ff1}.

\smallskip

A very important subscase is obtained when the constraints are
given in the form (Voronets-Chapliguin constraints type)
$
\dot{x}_\alpha-\Phi_\alpha\left(t,\textbf{x},\dot{x}_{M+1},\ldots,\dot{x}_{N}\right)=0,
$
for $\alpha=1,\ldots,M.$
 As we shall show  under these
assumptions the arbitrary functions are determined as follows:
$L_j=\dot{x}_j$ for $j=M+1,\ldots,N.$ Consequently the action of
the modified Lagrangian multipliers method and the action of the
classical Lagrangian multipliers  method are equivalently. In view
of \eqref{qpr} this equivalence always locally holds for any
constrained Lagrangian systems.
\end{remark}

\section{Proof of Theorems \ref{A}, \ref{A1} and \ref{A2}}

\begin{proof}[Proof of Theorem  \ref{A}]

In view of the equalities

\[\begin{array}{rl}
 \displaystyle \int_{t_0}^{t_1}\delta{L}\,dt =&\displaystyle
\int_{t_0}^{t_1}\displaystyle\sum_{k=1}^M\left(\dfrac{\partial
L}{\partial\lambda_k}\delta\lambda_k\right)dt+\displaystyle
\int_{t_0}^{t_1} \displaystyle\sum_{j=1}^N \left(\dfrac{\partial
L}{\partial x_j}\delta x_j+\dfrac{\partial L}{\partial
\dot{x}_j}\delta \dfrac{d{x}_j}{dt} \right)dt
\vspace{0.20cm}\\
=&\displaystyle
\int_{t_0}^{t_1}\displaystyle\sum_{k=1}^M\left(-L_k\delta\lambda_k\right)dt+\displaystyle
\int_{t_0}^{t_1}\displaystyle\sum_{j=1}^N\left(\dfrac{\partial
L}{\partial x_j}\delta x_j+\dfrac{\partial L}{\partial
\dot{x}_j}\dfrac{d}{dt}\delta{x_j} +\dfrac{\partial
L}{\partial\dot{x}_j}\left(\delta\dfrac{dx_j}{dt}-
\dfrac{d}{dt}\delta{x_j}\right)\right)dt\vspace{0.20cm}\\
=&\displaystyle
\int_{t_0}^{t_1}\displaystyle\sum_{k=1}^M\left(-L_k\delta\lambda_k\right)dt+\displaystyle
\int_{t_0}^{t_1}\displaystyle\sum_{j=1}^N\dfrac{d}{dt}\left(\dfrac{\partial
T}{\partial\dot{x}_j}\delta
x_j\right)dt\vspace{0.20cm}\\
&-\displaystyle
\int_{t_0}^{t_1}\displaystyle\sum_{j=1}^N\left(\left(-\dfrac{\partial
L}{\partial x_j}+\dfrac{d}{dt}\left(\dfrac{\partial L}{\partial
\dot{x}_j}\right)\right)\delta{x_j} +\dfrac{\partial
L}{\partial\dot{x}_j}\left(\delta\dfrac{dx_j}{dt}-
\dfrac{d}{dt}\delta{x_j}\right)\right)dt.
\end{array}
\]
Consequently

\[\begin{array}{rl}
\left.\displaystyle
\int_{t_0}^{t_1}\delta{L}\,dt\right|_{L_\nu=0}=&\displaystyle
\int_{t_0}^{t_1}\displaystyle\sum_{j=1}^N\left(\dfrac{d}{dt}\left(\dfrac{\partial
T}{\partial\dot{x}_j}\delta
x_j\right)-\left(E_jL-\displaystyle\sum_{k=1}^NA_{jk}\dfrac{\partial
L}{\partial\dot{x}_k}\right)\delta x_j\right)dt\vspace{0.20cm}\\
=&\displaystyle\sum_{j=1}^N\left.\dfrac{\partial
T}{\partial\dot{x}_j}\delta
x_j\right|_{t=t_0}^{t=t_1}-\displaystyle
\int_{t_0}^{t_1}\displaystyle\sum_{j=1}^N\left(E_jL-
\displaystyle\sum_{k=1}^NA_{jk}\dfrac{\partial
L}{\partial\dot{x}_k}\right)\delta x_jdt\vspace{0.20cm}\\
=&-\displaystyle
\int_{t_0}^{t_1}\displaystyle\sum_{j=1}^N\left(E_jL-
\displaystyle\sum_{k=1}^NA_{jk}\dfrac{\partial
L}{\partial\dot{x}_k}\right)\delta x_jdt=0,
\end{array}
\]where $\nu=1,\ldots,M.$
Here we use the equalities
$\delta{\textbf{x}}(t_0)=\delta{\textbf{x}}(t_1)=0.$
 Hence if
\eqref{24} holds then \eqref{023} is satisfied. The reciprocal
result is proved by choosing \[\delta x_k(t)=
  \begin{cases}
    \zeta (t)& \text{if}\,\, k=1, \\
     0& \text{otherwise},
  \end{cases}\]
  where $ \zeta (t)$ is a  positive function in the interval $(t^*_0,t^*_1),$
  and it is equal  to zero in the intervals
  $[t_0,\,t^*_0]$ and $[t^*_1,\,t_1],$ and applying Corollary \ref{Lagrange}.

{From} the definition \eqref{24} we have that
\[
D_\nu(fg)=D_\nu f\,g+f\,D_\nu\,g+\dfrac{\partial
f}{\partial\dot{x}_\nu}\dfrac{d g}{dt}+\dfrac{d
f}{dt}\dfrac{\partial g}{\partial\dot{x}_\nu},\quad D_\nu a=0,
\]where $a$ is a constant.

 Now we shall write \eqref{24} in a more convenient way
\[
\begin{array}{rl}
 0=D_\nu{L}=&D_\nu\left(L_0-
\displaystyle\sum_{j=1}^M\lambda_jL_j-\displaystyle\sum_{j=M+1}^N
\lambda^0_j
L_j\right)\vspace{0.20cm}\\
=&D_\nu\,L_0-\displaystyle\sum_{j=1}^MD_\nu\left(\lambda_jL_j\right)
-\displaystyle\sum_{j=M+1}^N
\lambda^0_j
D_\nu\,L_j\vspace{0.20cm}\\
=&D_\nu\,L_0-\displaystyle\sum_{j=M+1}^N \lambda^0_j
D_\nu\,L_j-\vspace{0.20cm}\\
&-\displaystyle\sum_{j=1}^M\left(D_\nu\,\lambda_j\,\,L_j+\lambda_jD_\nu\,L_j+\dfrac{d
\lambda_j}{dt}\dfrac{\partial L_j}{\partial\dot{x}_\nu}+\dfrac{d
L_j}{dt}\dfrac{\partial \lambda_j}{\partial\dot{x}_\nu}\right).
\end{array}
\]
{F}rom these relations and since the constraints $L_j=0$ for
$j=1,\ldots,M,$ we easily obtain equations \eqref{D24} or
equivalently \begin{equation} \label{25}
E_{\nu}L_0=\displaystyle\sum_{k=1}^NA_{jk}\dfrac{\partial
L_0}{\partial\dot{x}_k}+\sum_{j=1}^M\left(\lambda_{j}D_{\nu}L_{j}+
\dfrac{d\lambda_j}{dt}\dfrac{\partial{L_j}}{\partial{\dot{x}_\nu}}\right)
+\displaystyle\sum_{j=M+1}^N \lambda^0_j D_\nu\,L_j.
\end{equation}
Thus the theorem is proved.
\end{proof}

 Now we show that the differential
equations \eqref{25} for convenient functions $L_j$ constants
$\lambda^0_j$ for $j=M+1,\ldots,N$ and for convenient matrix $A$
describe the motion of the constrained Lagrangian systems.

 \begin{proof}[Proof of Theorem  \ref{A1}]
The matrix equation \eqref{a1} can be rewritten in components as
follows
\begin{equation}
\label{219}
\displaystyle\sum_{j=1}^NA_{kj}\dfrac{\partial
L_\alpha}{\partial\dot{x}_j}=E_kL_\alpha\Longleftrightarrow D_kL_\alpha=0,
 \end{equation}
 for $\alpha,\,k=1,\ldots,N.$
 Consequently the differential equations \eqref{25} become
\begin{equation}\label{TT1}
E_{\nu}L_0=\displaystyle\sum_{k=1}^N\left(A_{\nu
k}\dfrac{\partial L_0}{\partial\dot{x}_k}+
\dfrac{d\lambda_k}{dt}\dfrac{\partial{L_k}}{\partial{\dot{x}_\nu}}\right)
\Longleftrightarrow D_\nu
L_0=\displaystyle\sum_{j=1}^M\dfrac{d\lambda_j}{dt}\dfrac{\partial{L_j}}{\partial{\dot{x}_\nu}},
\end{equation}
which coincide with the first systems \eqref{210}.

In view of the condition $|W_1|\ne 0$ we can solve equation
\eqref{a1} with respect to $A$ and obtain $A=W^{-1}_1\Omega_1.$
Hence, by considering \eqref{219} we obtain the second systems
from \eqref{210} and the transpositional relation \eqref{211}.
\end{proof}

\begin{proof}[Proof of Theorem  \ref{A2}]
The matrix equation \eqref{a2} is equivalent to the systems
\[
\begin{array}{rl}
\displaystyle\sum_{j=1}^NA_{kj}\dfrac{\partial
L_\alpha}{\partial\dot{x}_j}=&E_kL_\alpha\Longleftrightarrow D_kL_\alpha=0,\vspace{0.2cm}\\
\displaystyle\sum_{j=1}^N{A_{kj}
\dfrac{\partial{L_0}}{\partial\dot{x}_j}}=&0,
\end{array}
\]for  $k=1,\ldots,N,$ and $\alpha=1,\ldots,N-1.$
 Thus, by considering that $\lambda^0_N=0$ we
deduce that systems \eqref{25} takes the form
\[
E_{\nu}L_0=\displaystyle\sum_{j=1}^M\dfrac{d\tilde{\lambda}_j}{dt}\dfrac{\partial{L_j}}{\partial{\dot{x}_\nu}}
.
\]
Hence we obtain systems \eqref{0210}. On the other hand from
\eqref{a2} we have that $A=W^{-1}_2\Omega_2.$ Hence we deduce that
the transpositional relation \eqref{23} can be rewritten in the
form \eqref{0211}.
\end{proof}
The mechanics basic on the Hamiltonian principle with non--zero
transpositional relations given by formula \eqref{23}, Lagrangian
\eqref{21} and equations of motion \eqref{24} are called here the
{\it modification of the vakonomic mechanics} and we shortly write
{\it MVM}.
\smallskip

From the proofs of Theorems \ref{A1} and \ref{A2} follows that the
relations \eqref{ff1} holds identically in MVM.
\begin{corollary}
Differential equations \eqref{210} are invariant under the change
\[L_0\longrightarrow L_0-\displaystyle\sum_{j=1}^Na_jL_j,\] where
the $a_j$'s  are constants for $j=1,\ldots,N.$
\end{corollary}
\begin{proof}
Indeed, from \eqref{TT1} and \eqref{219} it follows that
\[D_\nu
\left(L_0-\displaystyle\sum_{j=1}^Na_jL_j\right)=D_\nu\,
L_0-\displaystyle\sum_{j=1}^Na_jD_\nu\,L_j=
D_\nu\,L_0=\displaystyle\sum_{j=1}^M\dfrac{d{\lambda}_j}{dt}\dfrac{\partial{L_j}}{\partial{\dot{x}_\nu}}.\]
\end{proof}

\begin{remark} \label{R1}
The following interesting facts follow from Theorems \ref{A1} and
\ref{A2}.

\begin{itemize}
\item[(1)] The equations of motion obtained from Theorem \ref{A1} are
more general than the equations obtained from Theorem \ref{A2}.
Indeed in \eqref{210} there are $N-M$ arbitrary functions while in
\eqref{0210} are $N-M-1$ arbitrary functions.

\smallskip

\item[(2)]

{If the constraints are} linear in the velocity then between the
Lagrangian multipliers $\mu,\,\,\dfrac{d\lambda}{dt}$ and
$\dfrac{d\tilde{\lambda}}{dt}$ there is the following relation
\[
\mu=\dfrac{d\tilde{\lambda}}{dt}=\left(W^{-1}_2\right)^T\left(W^T_1
\dfrac{d\lambda}{dt}+W^{-1}_2\Omega^T_1W^{-T}_1\dfrac{\partial
L_0}{\partial\dot{\textbf{x}}}\right),
\]
where $W_1$ and $W_2$ are the  matrixes defined in Theorems \ref{A1}
and \ref{A2}.
\item[(3)]
{If the constraints are} linear in the velocity then one of the
important question which appear in MVM is related with the
arbitrariness functions $L_j$ for $j=M+1,\ldots,N.$ The following
question arise: Is it possible to  determine these functions in such
a way that $|W_1|$ or $|W_2|$  is non--zero everywhere in $\MM^*$?
If we have a positive answer to this question, then the equations of
motion of the MVM give a global behavior of the constrained
Lagrangian systems, i.e. the obtained motions completely coincide
with the motions obtained from the classical mathematical models.
Thus if  $|W_1|\ne 0$ and  $|W_2|\ne 0$  everywhere in $\MM^*$ then
we have the equivalence
\begin{equation}\label{HH}
D_{\nu}L_0={\displaystyle\sum_{j=1}^M
\dfrac{d\lambda_j}{dt}\dfrac{\partial{L_j}}{\partial{\dot{x}_\nu}}}
\Longleftrightarrow E_{\nu}L_0={\displaystyle\sum_{j=1}^M
\dfrac{d\tilde{\lambda}_j}{dt}\dfrac{\partial{L_j}}{\partial{\dot{x}_\nu}}}\Longleftrightarrow
E_{\nu}L_0=\displaystyle\sum_{j=1}^M
\mu_j\dfrac{\partial{L_j}}{\partial{\dot{x}_\nu}}
\end{equation}
\end{itemize}
\smallskip
{If the constraints are} nonlinear in the velocity and $|W_2|\ne 0$
everywhere in $\MM^*$ then we have the   equivalence
\begin{equation}\label{HHH}
E_{\nu}L_0={\displaystyle\sum_{j=1}^M
\dfrac{d\tilde{\lambda}_j}{dt}\dfrac{\partial{L_j}}{\partial{\dot{x}_\nu}}}\Longleftrightarrow
E_{\nu}L_0=\displaystyle\sum_{j=1}^M
\mu_j\dfrac{\partial{L_j}}{\partial{\dot{x}_\nu}}
\end{equation}
The equivalence with respect to the equations
$D_{\nu}L_0={\displaystyle\sum_{j=1}^M
\dfrac{d\lambda_j}{dt}\dfrac{\partial{L_j}}{\partial{\dot{x}_\nu}}}$
in general is not valid in this case because the term
$\Omega^T_1W^{-T}_1\dfrac{\partial L_0}{\partial\dot{\textbf{x}}}$
depend on $\ddot{\textbf{x}}.$

\end{remark}
\subsection{Application of Theorems \ref{A1} and \ref{A2} to
the Appell--Hamel mechanical systems}

\smallskip

 As a general rule the constraints studied in
classical mechanics are linear with respect to the velocities, i.e.
$L_j$ can be written as \eqref{1010}. However Appell and Hamel (see
\cite{Appell,Hamel}) in 1911, considered an artificial example of
nonlinear nonholonomic constrains.
 A big number of investigations
have been devoted to the derivation of the equations of motion of
mechanical systems with nonlinear nonholonomic constraints  see for
instance \cite{Chetaev,Hamel,NF,Novoselov}. The works of these
authors do not contain examples of systems with nonlinear
nonholonomic constraints differing essentially from the example
given by Appell and Hamel.

\smallskip
\begin{corollary}
 The equivalence \eqref{HH} also holds for the Appell -- Hamel system i.e. for the constrained Lagrangian systems
\[
\left(\mathbb{R}^3,\quad
\tilde{L}=\dfrac{1}{2}(\dot{x}^2+\dot{y}^2+\dot{z}^2)-gz,\quad
\{\dot{z}-a\sqrt{\dot{x}^2+\dot{y}^2}=0\}\right),
\]
where $a$ and $g$ are positive constants.
\end{corollary}
\smallskip
\begin{proof}
{The classical equations} \eqref{35} for the Appell-Hamel system are
\begin{equation}\label{313}
\ddot{x}=-\dfrac{a\dot{x}}{\sqrt{\dot{x}^2+\dot{y}^2}}\mu,\qquad
\ddot{y}=-\dfrac{a\dot{y}}{\sqrt{\dot{x}^2+\dot{y}^2}}\mu,\qquad\ddot{z}=-g+\mu,
\end{equation}
where $\mu$ is the Lagrangian multiplier.

\smallskip

Now we  apply Theorem \ref{A2}. Hence, in order to obtain that
$|W_2|\ne 0$ everywhere we choose the  functions  $L_j$ for
$j=1,2,3$ as follows
\[
 L_1=\dot{z}-a\sqrt{\dot{x}^2+\dot{y}^2}=0,\quad
L_2=\arctan\dfrac{\dot{x}}{\dot{y}}, \quad L_3=L_0=\tilde{L}.
\]

\smallskip

In this case the matrices $W_2,\,\Omega_2$ and $A$ are
 \[\begin{array}{rl}
W_2=& \left(\begin{array}{ccc}
-\dfrac{a\dot{x}}{\sqrt{\dot{x}^2+\dot{y}^2}}&-\dfrac{a\dot{y}}{\sqrt{\dot{x}^2+\dot{y}^2}}&1\\
\dfrac{\dot{y}}{\dot{x}^2+\dot{y}^2}&-\dfrac{\dot{x}}{\dot{x}^2+\dot{y}^2}&0\\
\dot{x}&\dot{y}&\dot{z}
\end{array}
\right),\quad |W_2|_{L_1=0}=1+a^2,\vspace{0.3cm}\\
 \Omega_2=& \left(\begin{array}{ccc}
-\dot{y}q&\dot{x}q&0\vspace{0.2cm}\\
\dfrac{\ddot{y}\left(\dot{x}^2-\dot{y}^2\right)-2\dot{x}\dot{y}\ddot{x}}{\left(\dot{x}^2+\dot{y}^2\right)^2}
&\dfrac{\ddot{x}\left(\dot{x}^2-\dot{y}^2\right)+2\dot{x}\dot{y}\ddot{y}}{\left(\dot{x}^2+\dot{y}^2\right)^2}&0\\
0&0&0
\end{array} \right),
\end{array}
\]
and the matrix $\left.A\right|_{L_1=0}$ is
\[\left(\begin{array}{ccc}
-\dfrac{\dot{y}\left(a^2\dot{y}\dot{x}\ddot{x}+\left((a^2+1)\dot{y}^2+\dot{x}^2\right)\ddot{y}\right)}
{(1+a^2)\left(\dot{x}^2+\dot{y}^2\right)}&
\dfrac{\left(a^2\dot{x}^2+(a^2+1)\left(\dot{y}^2+
\dot{x}^2\right)^2\right)\dot{y}\ddot{x}-a^2\dot{x}^3\ddot{y}}{(1+a^2)\left(\dot{x}^2+\dot{y}^2\right)^2}&0\\
\\
\dfrac{
\left(a^2\dot{y}^2+(a^2+1)\left(\dot{y}^2+
\dot{x}^2\right)\right)\dot{x}\ddot{y}-a^2\dot{y}^3\ddot{x}}{(1+a^2)\left(\dot{x}^2+\dot{y}^2\right)}&
-\dfrac{\dot{x}\left(a^2\dot{x}\dot{y}\ddot{y}+\left((a^2+1)\dot{x}^2+\dot{y}^2\right)\ddot{x}\right)}
{(1+a^2)\left(\dot{x}^2+\dot{y}^2\right)^2}&0\\
\\
\dfrac{\dot{y}a\left(\dot{y}\ddot{x}-\dot{x}\ddot{y}\right)}{(1+a^2)\left(\dot{x}^2+\dot{y}^2\right)^{3/2}}
&-\dfrac{\dot{x}a\left(\dot{y}\ddot{x}-\dot{x}\ddot{y}\right)}{(1+a^2)\left(\dot{x}^2+\dot{y}^2\right)^{3/2}}&0
\end{array}
\right).
\]
By considering that $|W_2|_{L_1=0}=1+a^2,$ we obtain that the
equations \eqref{0210} in this case describe the global behavior of
the Appell--Hamel systems and take the form

\smallskip

\begin{equation}\label{31133}
\ddot{x}=-\dfrac{a\dot{x}}{\sqrt{\dot{x}^2+\dot{y}^2}}\dot{\tilde{\lambda}},\qquad
\ddot{y}=-\dfrac{a\dot{y}}{\sqrt{\dot{x}^2+\dot{y}^2}}\dot{\tilde{\lambda}},\qquad\ddot{z}=-g+\dot{\tilde{\lambda}}.
\end{equation}
Clearly that this system coincide with classical differential
equations \eqref{313} with $\dot{\tilde{\lambda}}=\mu$.

\smallskip

 After the derivation of the constraint
$\dot{z}-a\,\sqrt{\dot{x}^2+\dot{y}^2}=0$ along the solutions of
\eqref{31133}, we obtain
\[0=\ddot{z}-a\dfrac{\ddot{x}}{\sqrt{\dot{x}^2+\dot{y}^2}}
+a\dfrac{\ddot{y}}{\sqrt{\dot{x}^2+\dot{y}^2}}=-g+(1+a^2)\dot{\tilde{\lambda}}.\]
Therefore $\dot\tilde{\lambda}= \dfrac{g}{1+a^2}.$ Hence the
equations of motion \eqref{31133} become
\begin{equation}\label{314}
\ddot{x}=-\frac{ag}{1+a^2}\frac{\dot{x}}{\sqrt{\dot{x}^2
+\dot{y}^2}},
\qquad\ddot{y}=-\frac{ag}{1+a^2}\frac{\dot{y}}{\sqrt{\dot{x}^2
+\dot{y}^2}},\qquad \ddot{z}=-\frac{a^2g}{1+a^2}.
\end{equation}
In this case the Lagrangian \eqref{b2}  writes
$$
L=\dfrac{1}{2}(\dot{x}^2+\dot{y}^2+\dot{z}^2)-
gz-\dfrac{g\,(t+C)}{1+a^2}(\dot{z}-a\sqrt{\dot{x}^2+\dot{y}^2})-\lambda^0_2\arctan\dfrac{\dot{x}}{\dot{y}},
$$
where $C$ and $\lambda^0_2$ are an arbitrary constants.

\smallskip

Under the condition $L_1=0$ we obtain that the transpositional
relations  are
 \begin{equation}\label{kat}
 \begin{array}{rl}
\delta\dfrac{dx}{dt}-\dfrac{d}{dt}\delta{x}=&
\dfrac{\dot{y}\left((1+a^2)\left(\dot{x}^2+\dot{y}^2\right)\left(\ddot{x}\delta{y}-\ddot{y}\delta{x}\right)
+a^2\dot{x}\left(\dot{y}\ddot{x}-\dot{x}\ddot{y}\right)\left(\dot{x}\delta{y}-\dot{y}\delta{x}\right)\right)}
{(1+a^2)\left(\dot{x}^2+\dot{y}^2\right)^2}
,\vspace{0.20cm}\\
\delta\dfrac{dy}{dt}-\dfrac{d}{dt}\delta{y}=&
\dfrac{\dot{x}\left((1+a^2)\left(\dot{x}^2+
\dot{y}^2\right)\left(\ddot{y}\delta{x}-\ddot{x}\delta{y}\right)
+a^2\dot{y}\left(\dot{y}\ddot{x}-\dot{x}\ddot{y}\right)\left(\dot{x}\delta{y}-
\dot{y}\delta{x}\right)\right)}{(1+a^2)\left(\dot{x}^2+\dot{y}^2\right)^2}
,\vspace{0.20cm}\\
\delta\dfrac{dz}{dt}-\dfrac{d}{dt}\delta{z}=&\dfrac{a\left(\dot{y}\ddot{x}-
\dot{x}\ddot{y}\right)\left(\dot{x}\delta{y}-
\dot{x}\delta{y}\right)}{(1+a^2)\left(\dot{x}^2+\dot{y}^2\right)^{3/2}}.
\end{array}
\end{equation}
{From} this example we obtain  that the independent virtual
variations $\delta x$ and $\delta y$  produce non--zero
transpositional relations. This result is not in accordance with
with the Suslov point on view on the transpositional relations.

\smallskip

  Now we apply Theorem \ref{A1}.  The
functions $L_0,\,L_1,\,L_2$ and $L_3$ are determined as  follows
\[
L_0=\tilde{L},\quad L_1=\dot{z}-a\,\sqrt{\dot{x}^2+\dot{y}^2},\quad
L_2=\dot{y},\quad L_3=\dot{x}.
\]
 Thus the matrix $W_1$ and $ \Omega_1$ are
\[
W_1=\left(\begin{array}{ccc}
-\dfrac{a\dot{x}}{\sqrt{\dot{x}^2+\dot{y}^2}}&-\dfrac{a\dot{y}}{\sqrt{\dot{x}^2+\dot{y}^2}}&1\\
0&1&0\\
1&0&0
\end{array}
\right),\quad \Omega_1=\left(\begin{array}{ccc}
\dot{y}q&-\dot{x}q&0\\
0&0&0\\
0&0&0\end{array} \right),
\]
where
$q=\dfrac{a(\ddot{x}\dot{y}-\ddot{x}\dot{y})}{\sqrt{\dot{x}^2+\dot{y}^2}^3}.$
Therefore $|W_1|=-1.$

 Hence, after some computations from \eqref{a1} we have that
\[
A=\left(\begin{array}{ccc}
0&0&0\\
0&0&0\\
\dot{y}q&-\dot{x}q&0\end{array} \right).
\]
 The equations of motion \eqref{210} becomes
\begin{equation}
\label{315}
\begin{array}{rl}
\ddot{x}=&-\dfrac{a^2\dot{y}}{{\dot{x}^2+\dot{y}^2}}(\dot{y}\ddot{x}
-\dot{x}\ddot{y})-\dfrac{a\dot{{\lambda}}}{\sqrt{\dot{x}^2+\dot{y}^2}}\dot{x},\vspace{0.2cm}\\
\ddot{y}=&-\dfrac{a^2\dot{x}}{{\dot{x}^2+\dot{y}^2}}(\dot{x}\ddot{y}-\dot{y}\ddot{x})
-\dfrac{a\dot{{\lambda}}}{\sqrt{\dot{x}^2+\dot{y}^2}}\dot{y},\vspace{0.2cm}\\
\ddot{z}=&-g+\dot{{\lambda}}.
\end{array}
\end{equation}
By solving these equations with respect to $\ddot{x},\,\ddot{y}$ and
$ \ddot{z}$ we obtain  the equations
\[
\ddot{x}=-\dfrac{a\dot{x}}{\sqrt{\dot{x}^2+\dot{y}^2}}\dot{{\lambda}},\qquad
\ddot{y}=-\dfrac{a\dot{y}}{\sqrt{\dot{x}^2+\dot{y}^2}}\dot{{\lambda}},\qquad\ddot{z}=-g+\dot{{\lambda}},
\]
 We observe in this case that $|W_1|=-1,$ consequently these
equations, obtained from Theorem \ref{A1}, give a global behavior of
the Appell--Hamel systems, i.e. coincide with the classical
equations \eqref{313} with
$\dot{{\lambda}}=\dot{{\tilde{\lambda}}}=\mu=\dfrac{g}{1+a^2}.$

\smallskip

The transpositional relations \eqref{211} can be written as
\begin{equation}\label{415}
\delta\dfrac{dx}{dt}-\dfrac{d}{dt}\delta\,x=0,
\quad\delta\dfrac{dy}{dt}-\dfrac{d}{dt}\delta\,y=0,\quad
\delta\dfrac{dz}{dt}-\dfrac{d}{dt}\delta\,
z=q\left(\dot{y}\delta\,x-\dot{x}\delta\,y\right).
\end{equation}
\end{proof}
{From} this corollary we observe that the independent virtual
variations $\delta x$ and $\delta y$  produce non--zero
transpositional relations \eqref{kat}  and zero transpositional
relations \eqref{415}.

\smallskip

The Lagrangian \eqref{b1} in this case takes the form
\[\begin{array}{rl}
L=&\dfrac{1}{2}(\dot{x}^2+\dot{y}^2+\dot{z}^2)-gz-
\dfrac{g\,(t+C)}{1+a^2}(\dot{z}-a\sqrt{\dot{x}^2+\dot{y}^2})-
\lambda^0_2\dot{y}-\lambda^0_3\dot{x}\vspace{0.2cm}\\
\simeq &\dfrac{1}{2}(\dot{x}^2+\dot{y}^2+\dot{z}^2)-gz-
\dfrac{g\,(t+C)}{1+a^2}(\dot{z}-a\sqrt{\dot{x}^2+\dot{y}^2}).
\end{array}
\]
 From  \eqref{Nat1} it
follows that
\[\delta\dfrac{dz}{dt}-\dfrac{d}{dt}\delta\,z=
q\left(\dot{y}\delta\,x-\dot{x}\delta\,y\right)+
\dfrac{a\dot{x}}{\sqrt{\dot{x}^2+\dot{y}^2}}\left(\delta\dfrac{dx}{dt}-
\dfrac{d}{dt}\delta\,x\right)+
\dfrac{a\dot{y}}{\sqrt{\dot{x}^2+\dot{y}^2}}\left(\delta\dfrac{dy}{dt}-
\dfrac{d}{dt}\delta\,y\right).
\]
Therefore this relation holds identically for \eqref{kat} and
\eqref{415}.

\smallskip

 In the next  sections we show the importance
of the equations of motion \eqref{210} and \eqref{0210} contrasting
them with the classical differential equations} of nonholonomic
mechanics.

\section{Modificated  vakonomic mechanics versus vakonomic
mechanics}
 Now we show that the  equations
of the {\it vakonomic mechanics} \eqref{K1} can be obtained  from
equations \eqref{D24}. More precisely, if in \eqref{23} we require
that all the virtual variations of the coordinates produce the zero
transpositional relations, i.e. the matrix $A$ is the zero matrix
and we require that $\lambda^0_j=0$ for $j=M+1,\ldots,N$, then from
\eqref{D24} by considering that $D_kL=E_kL,$ we obtain the vakonomic
equations \eqref{K1}, i.e.
\[
\begin{array}{rl} D_{\nu}
L_0=&\displaystyle\sum_{j=1}^M\left(\lambda_{j}D_{\nu} L_{j}+
\dfrac{d\lambda_j}{dt}\dfrac{\partial{L_j}}{\partial{\dot{x}_\nu}}\right)+\displaystyle\sum_{j=M+1}^N
\lambda^0_j D_\nu\,L_j{\Longrightarrow}\vspace{0.2cm}\\
  E_\nu\,L_0=&\displaystyle\sum_{j=1}^M\left(\lambda_jE_\nu\,L_j+
\dfrac{d\lambda_j}{dt}\dfrac{\partial{L_j}}{\partial
\dot{x}_\nu}\right),\quad
 {\nu=1,\ldots,N}
\end{array}\]

 In the following example in order to contrast Theorems \ref{A1}
  with the  vakonomic model we study the {\it skate or
knife edge on an inclined plane}.

 \smallskip

\textbf{Example 1. } To set up the problem, consider a plane $\Xi$
with cartesian
  coordinates $x$ and $y,$ slanted at an angle
$\alpha$. We assume that the $y$--axis is horizontal,
  while the $x$--axis is directed downward from the
  horizontal and let $(x,y)$ be the
  coordinates of the point of contact of the skate with the plane. The
angle $\varphi$ represents the orientation of the skate measured
from the $x$--axis. The skate is moving under the influence of the
gravity. Here the the acceleration due to gravity is denoted by
$g$. It also has mass $m,$ and the moment inertia of the skate
about a vertical axis through its contact point is denoted by $J,$
(see page 108 of \cite{NF} for a picture). The equation of
nonintegrable constraint is
\begin{equation}\label{skate}
L_1=\dot{x}\sin\varphi-\dot{y}\cos\varphi=0.
\end{equation}
With these notations the Lagrangian function of the skate is
\[
\hat{L}=\dfrac{m}{2}\left(\dot{x}^2+\dot{y}^2\right)+
\dfrac{J}{2}\dot\varphi^2+mg\,x\,\sin\alpha.
\]
 Thus we have the constrained mechanical systems
 \[
 \left( \mathbb{R}^2\times\mathbb{S}^1,\quad\hat{L}=\dfrac{m}{2}\left(\dot{x}^2+\dot{y}^2\right)+
\dfrac{J}{2}\dot\varphi^2+mg\,x\,\sin\alpha,\quad\{\dot{x}\sin\varphi-\dot{y}\cos\varphi=0\}\right).
\]
  For appropriate choice of mass, length and time units, we reduces the Lagrangian $\hat{L}$ to
  \[L_0=\dfrac{1}{2}\left(\dot{x}^2+\dot{y}^2+\dot{\varphi}^2\right)+x\,g\sin\alpha ,\]
  here for simplicity we leave the same notations for the all
  variables.
The question is, {\it what is the motion of the point of contact?}
To answer this question we shall use the vakonomic equations
\eqref{K1} and the equations \eqref{210}
 proposed in Theorem \ref{A1}.

\subsection{The study of the skate applying Theorem \ref{A1}}
  We determine the motion of the point of contact of the skate using Theorem \ref{A1}.
   We choose the arbitrary functions
  $L_2$ and $L_3$ as follows
  \[
    L_2=\dot{x}\cos\varphi+\dot{y}\sin\varphi,\quad
  L_3=\dot{\varphi},
  \]
   in order that the determinant
  $|W_1|\ne 0$ everywhere in the configuration space.

\smallskip

  The Lagrangian \eqref{b1}  becomes
  \[\begin{array}{rl}
  L(
  x,y,\varphi,\dot{x},\dot{y},\dot\varphi,\Lambda)=&\dfrac{1}{2}\left(\dot{x}^2
  +\dot{y}^2+\dot{\varphi}^2\right)+g\sin\alpha x
  -\lambda(\dot{x}\sin\varphi-\dot{y}\cos\varphi)-\lambda^0_3\dot{\varphi}\vspace{0.20cm}\\
  &\simeq \dfrac{1}{2}\left(\dot{x}^2+\dot{y}^2+\dot{\varphi}^2\right)+g\sin\alpha x
  -\lambda(\dot{x}\sin\varphi-\dot{y}\cos\varphi),
  \end{array}
  \]
  where $\lambda:=\lambda_1.$

  \smallskip

  The matrix $W_1$ and $\Omega_1$ are
\[\begin{array}{rl}
W_1=&\left(\begin{array}{cccc}
\sin\varphi&-\cos{\varphi}&0\\
\cos{\varphi}&\sin\varphi&0\\
0&0&1\\
\end{array}\right),\quad |W_1|=1,\vspace{0.2cm}\\
\Omega_1=&\left(\begin{array}{cccc}
\dot\varphi\cos{\varphi}&\dot\varphi\sin{\varphi}&-L_2\\
-\dot\varphi\sin{\varphi}&\dot\varphi\cos{\varphi}&-L_1\\
0&0&0\\
\end{array}\right).
\end{array}
\]
The matrix $A=W^{-1}_1\Omega_1$ becomes
\[A=\left.\left(\begin{array}{cccc}
0&\dot{\varphi}&-\sin\varphi L_2-\cos\varphi L_1\\
-\dot{\varphi}&0&\cos\varphi L_2-\sin\varphi L_1\\
0&0&0\\
\end{array}\right)\right|_{L_1=0}=\left(\begin{array}{cccc}
0&\dot{\varphi}&-\dot{y}\\
-\dot{\varphi}&0&\dot{x}\\
0&0&0\\
\end{array}\right).
\]
Hence the equation \eqref{210} and transpositional relations
\eqref{211} take the form
\begin{equation}\label{vvvv}
\ddot{x}+\dot{\varphi}\dot{y}=g\sin\alpha+\dot{\lambda}\sin\varphi,
\quad\ddot{y}-\dot{\varphi}\dot{x}=-\dot{\lambda}\cos\varphi,\quad
\ddot{\varphi}=0,
\end{equation}
and
\begin{equation}\label{T1}
\begin{array}{rl}
\delta\dfrac{dx}{dt}-\dfrac{d\delta x}{dt}=&\dot{y}\delta
\varphi-\dot{\varphi}\delta
y,\\
\delta\dfrac{dy}{dt}-\dfrac{d\delta
y}{dt}=&\dot{\varphi}\delta x-\dot{x}\delta{\varphi},\\
\delta\dfrac{d\varphi}{dt}-\dfrac{d\delta
\varphi}{dt}=&-L_2\left(\delta x\sin\varphi-\delta
y\cos\varphi\right)=0,
\end{array}
\end{equation}
respectively, here we have applied the Lagrange--Chetaev's
condition $ \sin\varphi\,\delta x-\cos\varphi\,\delta y=0. $

\smallskip

The initial conditions
\[
x_0=\left.x\right|_{t=0},\quad y_0=\left.y\right|_{t=0},\quad
\varphi_0=\left.\varphi\right|_{t=0},\quad
\dot{x}_0=\left.\dot{x}\right|_{t=0},\quad
\dot{y}_0=\left.\dot{y}\right|_{t=0},\quad
\dot{\varphi}_0=\left.\dot{\varphi}\right|_{t=0},
\]
satisfy the constraint, i.e.
\begin{equation}\label{Ic}
\sin\varphi_0\dot{x}_0-\cos\varphi_0\dot{y}_0=0.
\end{equation}
 After the derivation of the constraint along the solutions of the equation of motion \eqref{vvvv}, and using \eqref{skate} we  obtain
\[
\begin{array}{rl}0=&\sin\varphi\ddot x-\cos\varphi\ddot y+
\dot \varphi\left(\cos\varphi\dot
x+\sin\varphi\dot y\right)\vspace{0.2cm}\\
=&\sin\varphi\left(
g\sin\alpha+\dot{\lambda}\sin\varphi-\dot{\varphi}\dot{y}\right)-
\cos\varphi\left(-\dot{\lambda}\cos\varphi+\dot{\varphi}\dot{x}\right)+\dot
\varphi\left(\cos\varphi\dot x+\sin\varphi\dot y\right).
\end{array}
\]
Hence $\dot{\lambda}=-g\sin\alpha\sin\varphi.$ Therefore the
differential equations \eqref{vvvv} can be written as
\begin{equation}\label{LL000}
\ddot{x}+\dot{\varphi}\dot{y}=g\sin\alpha\cos^2\varphi,
\quad\ddot{x}-\dot{\varphi}\dot{x}=g\sin\alpha\sin\varphi\cos\varphi,\quad
\ddot{\varphi}=0. \end{equation}
 We study the motion of the skate in the following three cases:
\begin{itemize}
\item[(i)]$\left.\dot{\varphi}\right|_{t=0}=\omega=0.$
\item[(ii)]$\left.\dot{\varphi}\right|_{t=0}=\omega\ne 0.$
\item[(iii)]$\alpha=0.$
\end{itemize}
For the first case ($\omega=0$),  after the change of
variables
\[X=\cos\varphi_0\,x-\sin\varphi_0\,y,\quad
Y=\cos\varphi_0\,x+\sin\varphi_0\,y,\] the differential equations
\eqref{LL} and the constraint become
\[
\ddot{X}=0,\quad \ddot{Y}=g\sin\alpha\cos\varphi_0,\quad
\varphi=\varphi_0,\quad \dot{X}=0,
\]
respectively. Consequently
\[X=X_0,\quad
Y=g\sin\alpha\cos\varphi_0\dfrac{t^2}{2}+\dot{Y}_0t+Y_0,\quad
\varphi=\varphi_0,
\]
thus the trajectories are straight lines.

\smallskip

For the second  case ($\omega\ne 0$), we take
$\varphi_0=\dot{y}_0=\dot{x}_0=x_0=y_0=0$  in order to simplify
the computations. In view of the equality
$\dot{\varphi}=\left.\dot{\varphi}\right|_{t=0}=\omega$ and
denoting by $'$ the derivation with respect $\varphi$ we get that
\eqref{LL000} become
\begin{equation}\label{LL1}
x''+y'=\dfrac{g\sin\alpha}{\omega^2}\cos^2\varphi, \quad
x''-{x}'=\dfrac{g\sin\alpha}{\omega^2}\sin\varphi\cos\varphi,\quad
\varphi'=1.
\end{equation}
Which are easy to integrate and we obtain
\[
x=-\dfrac{g\sin\alpha}{4\omega^2}\cos{(2\varphi)},\quad
y=-\dfrac{g\sin\alpha}{4\omega^2}\sin{(2\varphi)}+\dfrac{g}{2\omega^2}\varphi,\quad
\varphi=\omega t,
\]
 which correspond to the equation of the cycloid. Hence the
point of contact of the skate follows a cycloid  along the plane,
but do not slide down the plane.

For the third case ($ \alpha=0$), if $\varphi_0=0,\,\omega\ne 0$
we obtain that the solutions of the given differential systems
\eqref{LL000}  are
\[x=\dot{y}_0\cos\varphi+\dot{x}_0\sin\varphi+a,\quad
y=\dot{y}_0\sin\varphi+\dot{y}_0\cos\varphi+b,\quad
\varphi=\varphi_0+\omega t,\] where
$a=x_0-\dfrac{\dot{y}_0}{\omega},\,b=y_0+\dfrac{\dot{x}_0}{\omega},$
which correspond to the equation of the circle with center at
$(a,b)$ and radius $\dfrac{\dot{x}^2_0+\dot{y}^2_0}{\omega^2}.$

If $ \alpha=0$ and $\varphi_0=0,\,\omega= 0$ then we obtain that
the solutions are
\[x=\dot{x}_0t+x_0,\quad y=\dot{y}_0t+y_0.\]
All these solutions coincide with the solutions obtained from the
Lagrangian equations \eqref{35} with multipliers  (see
\cite{Arnold})\[ \ddot{x}=g\sin\alpha+\mu\sin\varphi,
\quad\ddot{y}-\dot{\varphi}\dot{x}=-\mu\cos\varphi,\quad
\ddot{\varphi}=0,
\]
with $\mu=\dot{\lambda}=-g\sin\alpha\sin\varphi.$

\subsection{The study of the skate applying vakonomic model}
Now we consider instead of Theorem \ref{A1}  the vakomic model for
studying the motion of the skate.

We consider the Lagrangian
\[L(
  x,y,\varphi,\dot{x},\dot{y},\dot\varphi,\Lambda)=\dfrac{1}{2}\left(\dot{x}^2+\dot{y}^2+\dot{\varphi}^2\right)+g\,x\,\sin\alpha
  - \lambda(\dot{x}\sin\varphi-\dot{y}\cos\varphi).
  \]

The equations of motion \eqref{K1} for the skate are
\[\dfrac{d}{dt}\left(\dot{x}-\lambda\sin\varphi\right)=0,
\quad\dfrac{d}{dt}\left(\dot{y}+\lambda\cos\varphi\right)=0,\quad
\ddot{\varphi}=-\lambda\left(\dot{x}\cos\varphi+\dot{y}\sin\varphi\right).\]
We shall study only the case when $\alpha=0.$ After integration we
obtain the differential systems
\begin{equation}\label{mm}\begin{array}{rl}
\dot{x}=&\lambda\sin\varphi+a=\cos\varphi\left(a\cos\varphi+b\sin\varphi\right),
\\
\dot{y}=&-\lambda\cos\varphi+b=\sin\varphi\left(a\cos\varphi+b\sin\varphi\right),\\
\ddot{\varphi}=&\left(b\cos\varphi-
a\sin\varphi\right)\left(a\cos\varphi+b\sin\varphi\right)=(b^2_1+a^2_2)\sin(\varphi+\alpha)\cos(\varphi+\alpha),\\
\lambda=&b\cos\varphi-a\sin\varphi,
\end{array}
\end{equation}
where $a=\dot{x}_0-\lambda_0\sin\varphi_0,$
$b=\dot{y}_0+\lambda_0\cos\varphi_0$ and
$\lambda_0=\lambda|_{t=0}$ is an arbitrary parameter. After the
integration of the third equation we obtain that
\begin{equation}\label{1mm}
\displaystyle\int_{0}^{\varphi}\dfrac{d\varphi}{
\sqrt{1-\kappa^2\sin^2\varphi}}=t\sqrt{\dfrac{h+a^2+b^2}{2}},
\end{equation}
where $h$ is an arbitrary constant which we choose in such a way
that $\kappa^2=\dfrac{2(a^2+b^2)}{h+a^2+b^2}<1.$

{F}rom \eqref{1mm} we get
$\sin\varphi=sn\left(t\sqrt{\dfrac{h+a^2+b^2}{2}}\right),\quad
\cos\varphi=cn\left(t\sqrt{\dfrac{h+a^2+b^2}{2}}\right),$ where
$sn$ and $cn$ are the Jacobi elliptic functions . Hence, if we
take $\dot{x}_0=1,\,\dot{y}_0=\varphi_0=0,$ then the solutions of
the differential equations \eqref{mm} are
\begin{equation}\label{KK1}
\begin{array}{rl}
x=&x_0+\displaystyle\int_{t_0}^t\left(
cn\left(t\sqrt{\dfrac{h+1+\lambda^2_0}{2}}\right)sn\left(t\sqrt{\dfrac{h+1+\lambda^2_0}{2}}\right)+
\lambda_0sn\left(t\sqrt{\dfrac{h+1+\lambda^2_0}{2}}\right)\right)dt,\vspace{0.2cm}\\
y=&y_0+\displaystyle\int_{t_0}^t
sn\left(t\sqrt{\dfrac{h+1+\lambda^2_0}{2}}\right)\,\lambda_0\,sn\left(t\sqrt{\dfrac{h+1+\lambda^2_0}{2}}\right)dt,\vspace{0.2cm}\\
\varphi=&am\left(t\sqrt{\dfrac{h+1+\lambda^2_0}{2}}\right).
\end{array}
\end{equation}
It is interesting to compare this amazing  motions  with the
motions that we obtained above. For the same initial conditions
the skate moves sideways along the circles.
 By considering that the solutions \eqref{KK1} depend on the
 arbitrary parameter $\lambda_0$ we obtain that for the given
 initial conditions do not exist a unique solution of the
 differential equations in the vakonomic model. Consequently the
 principle of determinacy is not valid for vakonomic mechanics
 with nonintegrable constraints
 (see the Corollary of page 36 in \cite{Arnold}).

\section{Modificated  vakonomic mechanics versus Lagrangian
 and constrained Lagrangian mechanics}

\subsection{MVM versus Lagrangian mechanics}

The Lagrangian equations which describe the motion of the
Lagrangian systems can be obtained from Theorem \ref{A1}
 by supposing that $M=0,$ i.e. there is  no constraints We
choose the arbitrary functions $L_\alpha$ for $\alpha=1,\ldots,N$ as
follows
$$
\quad L_\alpha=\dfrac{d x_\alpha }{dt},\quad \alpha=1,\ldots,N.
$$
 Hence the Lagrangian  \eqref{b1}  takes the form
$$
L=L_0-\sum_{j=1}^N\lambda^0_j\dfrac{d x_j}{dt}\simeq L_0.
$$
In this case we have that  $|W_1|=1.$

By considering the property of the Lagrangian derivative (see
\eqref{An1})
 we obtain that  $\Omega_1$ is a zero matrix . Hence the matrices  $A_1$ is the zero
 matrix.    As a consequence  the equations
\eqref{210} become
 \[D_{\nu}L=E_\nu L= E_\nu\left(
L_0-\displaystyle\sum_{j=1}^N\lambda^0_j\dot{x}_j\right)= E\nu
L_0=0\]
 because $L\simeq L_0.$
The transpositional relation \eqref{211} in this case are
$\delta\dfrac{d\textbf{x}}{dt}-\dfrac{d
\delta{\textbf{x}}}{dt}=0,$ which are the well known relations in
the Lagrangian  mechanics  (see formula \eqref{1313}).
\subsection{MVM versus constrained Lagrangian systems}

\smallskip

{From} the equivalences \eqref{HH} we have that in the case when the
constraints are linear in the velocity the equations of motions of
the MVM coincide with the Lagrangian equations with multipliers
\eqref{35}
 except perhaps in a
zero Lebesgue measure set   $|W_2|=0$ or $|W_1|=0.$ When the
constraints are nonlinear in the velocity, we have the equivalence
\eqref{HHH}. Consequently equations of motions of the MVM coincide
with the Lagrangian equations with multipliers \eqref{35}
 except perhaps in a zero Lebesgue measure set   $|W_2|=0.$

We illustrate this result in the following example.

 \textbf{Example 2.} Let
\[\left(\mathbb{R}^2,\quad L_0=\dfrac{1}{2}\left(\dot{x}^2+\dot{y}^2\right)-U(x,y),\quad
\{2\left({x}\dot{{x}}+{y}\dot{{y}}\right)=0\}\right),\]
 be the
constrained Lagrangian systems.

 In order to apply  Theorem
\ref{A1} we choose the arbitrary function $L_1$ and $L_2$ as follow
\begin{itemize}
\item[(a)]
\[ L_1=2\left({x}\dot{{x}}+{y}\dot{{y}}\right),\quad L_2=-y\dot{x}+x\dot{y}.\] Thus the matrices $W_1$ and $\Omega_1$
are
\[W_1=\left(
\begin{array}{cc}
2x&2y\\
-y&x\\
\end{array}\right),\quad |W_1|=2x^2+2y^2=2,\quad\Omega_1=\left(
\begin{array}{cc}
0&0\\
-2\dot{y}&2\dot{x}\\
\end{array}\right).\]
Consequently equations \eqref{210} describe the motion everywhere
for the constrained Lagrangian systems.

Equations \eqref{210}  become
\[\begin{array}{rl}
\ddot{x}=\left.-\dfrac{\partial U}{\partial
x}+2\dot{y}\left(y\dot{x}-x\dot{y}\right)+2x\dot\lambda\right|_{L_1=0}=-\dfrac{\partial
U}{\partial
x}+x\left(\dot\lambda-2(\dot{x}^2+\dot{y}^2)\right),\vspace{0.2cm}\\
\ddot{y}=\left.-\dfrac{\partial U}{\partial
y}-2\dot{x}\left(y\dot{x}-x\dot{y}\right)+2y\dot\lambda\right|_{L_1=0}=-\dfrac{\partial
U}{\partial y}+y\left(\dot\lambda-2(\dot{x}^2+\dot{y}^2)\right),
\end{array}
\]

 Transpositional relations take the form
\begin{equation}\label{m2}
\delta\dfrac{d{x}}{dt}-\dfrac{d\delta{x}}{dt}=2y\left(\dot{y}\delta
x-\dot{x}\delta y\right),\quad
\delta\dfrac{d{y}}{dt}-\dfrac{d\delta{y}}{dt}=-2x\left(\dot{y}\delta
x-\dot{x}\delta y\right).
\end{equation}
\item[(b)]
If we choose
$L_2=\dfrac{y\dot{x}}{x^2+y^2}-\dfrac{x\dot{y}}{x^2+y^2}=\dfrac{d}{dt}\arctan{\dfrac{x}{y}},$
then
\[W_1=\left(
\begin{array}{cc}
2x&2y\\
\dfrac{y}{x^2+y^2}&-\dfrac{x}{x^2+y^2}\\
\end{array}\right),\quad |W_1|=-2,\quad\Omega_1=\left(
\begin{array}{cc}
0&0\\
0&0\\
\end{array}\right).\]
Equations \eqref{210}  and transpositional relations become
\[
\ddot{x}=-\dfrac{\partial
U}{\partial\,x}+2x\dot\lambda,\quad\ddot{y}=-\dfrac{\partial
U}{\partial y}+2y\dot\lambda,
\]
\begin{equation}\label{m22}
\delta\dfrac{d{x}}{dt}-\dfrac{d\delta{x}}{dt}=0,\quad
\delta\dfrac{d{y}}{dt}-\dfrac{d\delta{y}}{dt}=0.
\end{equation}
respectively.
\end{itemize}
{F}rom this example we obtain that for the holonomic constrained
Lagrangian systems the transpositional relations can be non--zero
(see \eqref{m2}), or can be zero (see \eqref{m22}). We observe that
from condition \eqref{Nat1} it follows the relation
\[x\left(\delta\dfrac{d{x}}{dt}-\dfrac{d\delta{x}}{dt}\right)+
y\left(\delta\dfrac{d{y}}{dt}-\dfrac{d\delta{y}}{dt}\right)=0.\]
This equality holds identically if \eqref{m22} and \eqref{m2} takes
place.

The equations of motions \eqref{35} in this case are
\[
\ddot{x}=-\dfrac{\partial
U}{\partial\,x}+2x\,\mu,\quad\ddot{y}=-\dfrac{\partial U}{\partial
y}+2y\,\mu,
\]
with $\mu=\dot\lambda-2(\dot{x}^2+\dot{y}^2).$

\smallskip

\textbf{Example 3.}  To contrast the MVM with the classical model we
apply Theorems \ref{A1}  to the  {\it Gantmacher's systems}
{\rm(}see for more details \cite{Gant,Sad}{\rm )}.

\smallskip

 Two material points $m_1$ and $m_2$ with equal masses are linked by a
metal rod with fixed length $l$  and small mass. The systems can
move only in the vertical plane and so the speed of the midpoint
of the rod is directed along the rod. It is necessary to determine
the trajectories of the material points $m_1$ and $m_2.$

 \smallskip

 Let $(q_1,\,r_1)$ and $(q_2,\,r_2)$ be the coordinates of the
points  $m_1$ and $m_2,$ respectively. Clearly
$(q_1-q_2)^2+(r_1-r_2)^2=l^2.$ Thus we have a  constrained Lagrangian system in the
configuration space $ \mathbb{R}^4$ with the Lagrangian function
$\textsc{L}=\dfrac{1}{2}\left(\dot{q}^2_1+
\dot{q}^2_2+\dot{r}^2_1+\dot{r}^2_2\right)-g/2{(r_1+r_2)},$ and
with the linear constraints
\[(q_2-q_1)(\dot{q}_2-\dot{q}_1)+(r_2-r_1)(\dot{r}_2-\dot{r}_1)=0,\quad
(q_2-q_1)(\dot{r}_2+\dot{r}_1)-(r_2-r_1)(\dot{q}_2+\dot{q}_1)=0.
\]
 Introducing the following
change of coordinates:
\[
 x_1=\dfrac{q_2-q_1}{2},\quad
x_2=\dfrac{r_1-r_2}{2}, \quad x_3=\dfrac{r_2+r_1}{2},\quad
x_4=\dfrac{q_1+q_2}{2},
\]
 we obtain
$x^2_1+x^2_2=\dfrac{1}{4}\left((q_1-q_2)^2+(r_1-r_2)^2\right)=\dfrac{l^2}{4}.$
Hence we have the  constrained Lagrangian mechanical systems
\[
\left(\mathbb{R}^4,\quad
\tilde{L}=\displaystyle\frac{1}{2}\sum_{j=1}^4\dot{x}^2_j-gx_3,\quad
\{x_1\dot{x}_1+x_2\dot{x}_2=0,\quad
x_1\dot{x}_3-x_2\dot{x}_4=0\}\right).
\]
  The equations of motion \eqref{35}
obtained from the d'Alembert--Lagrange principle are
\begin{equation}\label{G010}
\ddot{x}_1=\mu_1x_1,\quad\ddot{x}_2=\mu_1x_2,\quad\ddot{x}_3=-g+\mu_2x_1,\quad\ddot{x}_4=-\mu_2x_2,
\end{equation}where $\mu_1,\,\mu_2$ are the Lagrangian multipliers
such that
\begin{equation}\label{Gg31}
\mu_1=-\dfrac{\dot{x}^2_1+\dot{x}^2_2}{x^2_1+x^2_2},\quad
\mu_2=\dfrac{
\dot{x}_2\dot{x}_4-\dot{x}_1\dot{x}_3+gx_1}{x^2_1+x^2_2}.
\end{equation}
 For applying Theorem \ref{A1} we have  the constraints
\[
L_1=x_1\dot{x}_1+x_2\dot{x}_2=0,\quad
L_2=x_1\dot{x}_3-x_2\dot{x}_4=0,\] and we choose the arbitrary
functions $L_3$ and $L_4$ as follows
\[
L_3=-x_1\dot{x}_2+x_2\dot{x}_1,\quad
L_4=x_2\dot{x}_3+x_1\dot{x}_4.
  \]
 For the given functions we
 obtain that
\[
W_1=\left(\begin{array}{cccc}
x_1&x_2&0&0\\
0&0&x_1&-x_2\\
x_2&-x_1&0&0\\
0&0&x_2&x_1
\end{array}
\right),\quad \Omega_1=\left(\begin{array}{cccc}
0&0&0&0\\
-\dot{x}_3&\dot{x}_4&\dot{x}_1&-\dot{x}_2\\
-2\dot{x}_2 &2\dot{x}_1&0&0\\
-\dot{x}_4&-\dot{x}_3&\dot{x}_2&\dot{x}_1
\end{array} \right).
\]
Therefore $|W_1|=(x^2_1+x^2_2)^2=\dfrac{l^2}{4}\ne 0.$ The matrix
$A$ in this case is
\[
\left(\begin{array}{cccc}
\dfrac{2x_2\dot{x}_2}{x^2_1+x^2_2}&-\dfrac{2x_2\dot{x}_1}{x^2_1+x^2_2}
&0&0\vspace{0.2cm}\\
-\dfrac{2x_1\dot{x}_2}{x^2_1+x^2_2}&\dfrac{2x_1\dot{x}_1}{x^2_1+x^2_2}
&0&0\vspace{0.2cm}\\
-\dfrac{x_1\dot{x}_3+x_2\dot{x}_4}{x^2_1+x^2_2}&\dfrac{x_1\dot{x}_4-
x_2\dot{x}_3}{x^2_1+x^2_2}
&\dfrac{x_1\dot{x}_1+x_2\dot{x}_2}{x^2_1+x^2_2}&\dfrac{x_2\dot{x}_1-
x_1\dot{x}_2}{x^2_1+x^2_2}\vspace{0.2cm}\\
\dfrac{x_1\dot{x}_4-x_2\dot{x}_3}{x^2_1+x^2_2}&\dfrac{x_1\dot{x}_3-x_2\dot{x}_4}{x^2_1
+x^2_2}&\dfrac{x_2\dot{x}_1-x_1\dot{x}_2}{x^2_1+x^2_2}
&\dfrac{x_1\dot{x}_1+x_2\dot{x}_2}{x^2_1+x^2_2}
\end{array}
\right).
\]

Consequently differential equations \eqref{210} take the form
\begin{equation}\label{lll}
\begin{array}{rl}
\ddot{x}_1=&\left.\left(\dfrac{2x_2\dot{x}_1\dot{x}_2-2x_1\dot{x}^2_2-x_1\dot{x}^2_3-
x_1\dot{x}^2_4}{x^2_1+x^2_2}+x_1\dot{\lambda}_1\right)\right|_{L_1=L_2=0}\vspace{0.2cm}\\
=&x_1\left(\dot{\lambda}_1-
\dfrac{2\dot{x}^2_1+2\dot{x}^2_2+\dot{x}^2_3+\dot{x}^2_4}{x^2_1+
x^2_2}\right),\vspace{0.2cm}\\
\ddot{x}_2=&-\left.\left(\dfrac{-2x_1\dot{x}_1\dot{x}_2+2x_2\dot{x}^2_2+x_2\dot{x}^2_3+
x_2\dot{x}^2_4}{x^2_1+x^2_2}+x_2\dot{\lambda}_1\right)\right|_{L_1=L_2=0}\vspace{0.2cm}\\
=&x_2\left(\dot{\lambda}_1-
\dfrac{2\dot{x}^2_1+2\dot{x}^2_2+\dot{x}^2_3+
\dot{x}^2_4}{x^2_1+x^2_2}\right),\vspace{0.2cm}\\
\ddot{x}_3=&\left.\left(\dfrac{\dot{x}_3\left(x_1\dot{x}_1+x_2\dot{x}_2\right)-
\dot{x}_4\left(x_2\dot{x}_1-x_1\dot{x}_2\right)}{x^2_1+x^2_2}
+x_1\dot{\lambda}_2-g\right)\right|_{L_1=L_2=0}\vspace{0.2cm}\\
=&\dfrac{\dot{x}_4\left(x_2\dot{x}_1
-x_1\dot{x}_2\right)}{x^2_1+x^2_2}
+x_1\dot{\lambda}_2-g,\vspace{0.2cm}\\
\ddot{x}_4=&\left.\left(\dfrac{\dot{x}_4\left(x_1\dot{x}_1+x_2\dot{x}_2\right)
-\dot{x}_3\left(x_2\dot{x}_1-x_1\dot{x}_2\right)}{x^2_1+x^2_2}
-x_2\dot{\lambda}_2\right)\right|_{L_1=L_2=0}\vspace{0.2cm}\\
=&-\dfrac{\dot{x}_3\left(x_2\dot{x}_1
-x_1\dot{x}_2\right)}{x^2_1+x^2_2} -x_2\dot{\lambda}_2.
\end{array}
 \end{equation}
 Derivating the constraints we obtain  that the multipliers $\dot{\lambda}_1$ and
 $\dot{\lambda}_2$ are
 \[
 \dot{\lambda}_1=\dfrac{\dot{x}^2_1+\dot{x}^2_2+\dot{x}^2_3+\dot{x}^2_4}{x^2_1
 +x^2_2}=\mu_1+\dfrac{\dot{x}^2_3+\dot{x}^2_4}{x^2_1
 +x^2_2},\quad\dot{\lambda}_2=\dfrac{gx_1}{x^2_1+x^2_2}=\mu_2+\dfrac{
\dot{x}_1\dot{x}_3-\dot{x}_2\dot{x}_4}{x^2_1+x^2_2}.
 \]
 Inserting these values into \eqref{lll} we deduce
\[
\begin{array}{ll}
\ddot{x}_1=&-\dfrac{x_1\left(\dot{x}^2_1+\dot{x}^2_2\right)}{x^2_1+x^2_2},\qquad
\ddot{x}_2=\dfrac{x_2\left(\dot{x}^2_1+\dot{x}^2_2\right)}{x^2_1+x^2_2},\vspace{0.2cm}\\
\ddot{x}_3=&-g+\dfrac{x_1\left(
\dot{x}_2\dot{x}_4-\dot{x}_1\dot{x}_3+gx_1\right)}{x^2_1+x^2_2},\quad
\ddot{x}_4=-\dfrac{x_2\left(
\dot{x}_2\dot{x}_4-\dot{x}_1\dot{x}_3+gx_1\right)}{x^2_1+x^2_2}.
\end{array}
\]
These equations coincide with equations \eqref{G010} everywhere
because $|W_1|=\dfrac{l^2}{4},$ where $l$ is the length of the
rod.

\smallskip

 The transpositional relations in this case are

\begin{equation}\label{Llll}
\begin{array}{ll}
\delta\dfrac{d{x}_1}{dt}-\dfrac{d\delta{x}_1}{dt}=&-\dfrac{2x_2}{x^2_1+x^2_2}\left(\dot{x}_1\delta
x_2- \dot{x}_2\delta x_1\right),\vspace{0.2cm}\\
\delta\dfrac{d{x}_2}{dt}-\dfrac{d\delta{x}_2}{dt}=&\dfrac{2x_1}{x^2_1+x^2_2}\left(\dot{x}_1\delta
x_2- \dot{x}_2\delta x_1\right),\vspace{0.2cm}\\
\delta\dfrac{d{x}_3}{dt}-\dfrac{d\delta{x}_3}{dt}=&\dfrac{x_1}{x^2_1+x^2_2}\left(\dot{x}_1\delta
x_3- \dot{x}_3\delta x_1 +\dot{x}_4\delta x_2- \dot{x}_2\delta
x_4\right),\vspace{0.2cm}\\
&+\dfrac{x_2}{x^2_1+x^2_2}\left(\dot{x}_1\delta x_4-
\dot{x}_4\delta x_1+\dot{x}_2\delta x_3- \dot{x}_3\delta
x_2\right),\\
\delta\dfrac{d{x}_4}{dt}-\dfrac{d\delta{x}_4}{dt}=&-\dfrac{x_2}{x^2_1+x^2_2}\left(\dot{x}_1\delta
x_3- \dot{x}_3\delta x_1+\dot{x}_4\delta x_2- \dot{x}_2\delta
x_4\right)\vspace{0.3cm}\\
&+\dfrac{x_1}{x^2_1+x^2_2}\left(\dot{x}_1\delta x_4-
\dot{x}_4\delta x_1+\dot{x}_2\delta x_3- \dot{x}_3\delta
x_2\right).
\end{array}
  \end{equation}

From this example we again get that  the  virtual variations produce
the non--zero transpositional relations.

\begin{remark}
From the previous example we observe that  the virtual
variations produce zero or non--zero transpositional relations,
depending on the arbitrary functions which appear in the
construction of the proposed mathematical model. Thus, the
following question arises: Can  be choosen the arbitrary functions
$L_j$ for $j=M+1,\ldots,N$ in such a way that for the nonholonomic
systems only the independent virtual variations would generate
zero transpositional relations?
\end{remark}
The positive answer to this question {is obtained locally
for any constrained Lagrangian systems and globally for}  the {\it
Chaplygin-Voronets mechanical systems,} and for the generalization
of these systems studied in the next section.

\section{ MVM and nonholonomic generalized Voronets--Chaplygin systems.
 Proofs of  Theorem \ref{A0} and Proposition \ref{aa1} and \ref{aa2}.}

It was pointed out by Chaplygin \cite{Chaplygin} that in many
conservative nonholonomic systems the generalized coordinates
$$\left(\textbf{x},\textbf{y}\right):=
\left(x_1,\ldots,x_{s_1},y_1,\ldots,y_{s_2}\right),\quad
s_1+s_2=N,$$ can be chosen in such a way that the Lagrangian
function and the constraints take the simplest form. In particular
Voronets in \cite{Voronets} studied the constrained Lagrangian
systems with Lagrangian
$\tilde{L}=\tilde{L}\left(\textbf{x},\textbf{y},\dot{\textbf{x}},\dot{\textbf{y}}\right)$
and constraints \eqref{Vor}. This systems is called the {\it
Voronets mechanical systems}.

 We shall apply equations \eqref{210} to study
the generalization of the Voronets systems, which we define now.

\smallskip

The constrained Lagrangian mechanical systems
\begin{equation}\label{Ra1}
\left(\textsc{Q},\quad
\tilde{L}\left(t,\textbf{x},\textbf{y},\dot{\textbf{x}},\dot{\textbf{y}}\right)
,\quad \{\dot{x}_\alpha-\Phi_\alpha
\left(t,\textbf{x},\textbf{y},\dot{\textbf{y}}\right)=0,\quad
\alpha=1,\ldots,s_1\}\right),
\end{equation}
is called the {\it  generalized Voronets mechanical systems.}

\smallskip

An example of generalized Voronets systems is Appell-Hamel systems
analyzed in the previous subsection.

\smallskip

\begin{corollary} \label{BBB}
Every Nonholonomic  constrained Lagrangian mechanical systems
locally is a generalized Voronets mechanical systems.
\end{corollary}
\begin{proof}
Indeed, the independent constraints can be locally represented in
the form \eqref{qpr}. Thus by introducing the coordinates
\[ x_j=x_j,\quad \mbox{for}\quad j=1,\ldots,M,\quad x_{M+k}=y_k,\quad \mbox{for}\quad k=1,\ldots, N-M,\]
then we have that any constrained Lagrangian mechanical systems is
 locally a generalized Voronets mechanical systems.
\end{proof}

\begin{proof}[Proof of Theorem  \ref{A0}] For simplicity we shall study only  scleronomic  generalized
Voronets systems.

\smallskip

 To determine equations \eqref{210} we suppose that
\begin{equation}\label{Ch22}
L_\alpha=\dot{x}_\alpha-\Phi_\alpha \left(
\textbf{x},\textbf{y},\dot{\textbf{y}}\right)=0,\quad
\alpha=1,\ldots,s_1.
\end{equation}
 It is evident from
the form of the constraint equations that the virtual variations
$\delta{\textbf{y}},$ are independent by definition. The
remaining variations $\delta{\textbf{x}},$ can be expressed in
terms of them by the relations (Chetaev's conditions)
\begin{equation}\label{44}
\delta{x}_\alpha-
\sum_{j=1}^{s_2}\frac{\partial{L_\alpha}}{\partial{\dot{y_j}}}\delta{y_j}=0,\quad
\alpha=1,\ldots,s_1.
\end{equation}
We shall apply Theorem \ref{A1}.  To construct the matrix $W_1.$ We
first determine $L_{{s_1}+1},\ldots,L_{s_1+s_2}=L_N$ as follow:
$$L_{s_1+j}=\dot{y}_j,\quad
j=1,\ldots,s_2.$$Hence, the  Lagrangian \eqref{21} becomes
\begin{equation}\label{Ra2}
L=L_0-\sum_{j=1}^{s_1}\lambda_j\left(\dot{x}_\alpha-\Phi_\alpha
(x,y,\dot{y})\right)-\sum_{j=s_1+1}^{N}\lambda^0_j\dot{y}_j\simeq
L_0-\sum_{j=1}^{s_1}\lambda_j\left(\dot{x}_\alpha-\Phi_\alpha
(x,y,\dot{y})\right).
\end{equation}
 The matrices $W_1$ and $W^{-1}_1$ are
\begin{equation}\label{45}
\left(
\begin{array}{ccccccc}
1&\hdots&0&0&a_{11}&\hdots&a_{{s_2}1}\\
0&\hdots&0&0&a_{12}&\hdots&a_{{s_2}2}\\
\vdots&\hdots&\vdots&\vdots&\vdots&\hdots&\vdots\\
0&\hdots&\vdots&1&a_{1{s_1}}&\hdots&a_{{s_2}{s_1}}\\
0&\hdots&0&0&1&\hdots&0\\
\vdots&\hdots&\vdots&\vdots&\vdots&\hdots&\vdots\\
0&\hdots&0&0&0&\hdots&1\end{array} \right),\quad\left(
\begin{array}{ccccccc}
1&\hdots&0&0&-a_{11}&\hdots&-a_{{s_2}1}\\
0&\hdots&0&0&-a_{12}&\hdots&-a_{{s_2}2}\\
\vdots&\hdots&\vdots&\vdots&\vdots&\hdots&\vdots\\
0&\hdots&\vdots&1&a_{1{s_1}}&\hdots&-a_{{s_2}{s_1}}\\
0&\hdots&0&0&1&\hdots&0\\
\vdots&\hdots&\vdots&\vdots&\vdots&\hdots&\vdots\\
0&\hdots&0&0&0&\hdots&1\end{array} \right),
\end{equation}
respectively, where $
a_{\alpha\,j}=\dfrac{\partial{L_\alpha}}{\partial\dot{y}_j},$ and the
matrices $\Omega_1$  and $A$ are
\begin{equation}\label{46}
 A=\Omega_1:=
 \left(\begin{array}{ccccccc}
E_1(L_1)&\hdots&E_{s_1}(L_{1})&E_{s_1+1}(L_{1})&\hdots&E_N(L_1)\\
\vdots&\hdots&\vdots&\hdots&\hdots&\vdots\\
E_1(L_{s_1})&\hdots&E_{s_1}(L_{s_1})&E_{s_1+1}(L_{s_1})&\hdots&E_N(L_{s_1})\\
0&\hdots&0&\hdots&0&0\\
\vdots&\hdots&\vdots&\hdots&\hdots&\vdots\\
0&\hdots&0&\hdots&0&0\end{array} \right),
\end{equation}
respectively. Consequently the differential equations \eqref{210} take the form \eqref{48}.

\smallskip

The transpositional relations \eqref{211} in view of \eqref{44} take the form \eqref{49}.
 As we can observe from \eqref{49} the independent virtual
variations $\delta{\textbf{y}}$ for the systems with the
constraints \eqref{Ch22} produce the zero transpositional
relations. The fact that the transpositional relations are zero follows
automatically and it is not necessary to assume it a priori, and it
is valid in general for the constraints which are nonlinear in
the velocity variables.

We observe that the relations \eqref{Nat1} in this case take the
form
\[\delta\dfrac{dx_\alpha}{dt}-\dfrac{d}{dt}\delta\,x_\alpha+\displaystyle\sum_{m=1}^{s_2}\dfrac{\partial
L_\alpha}{\partial \dot{y}_m}
\left(\delta\dfrac{dy_m}{dt}-\dfrac{d}{dt}\delta\,y_m\right)=\displaystyle\sum_{k=1}^{s_1}
E_k(L_\alpha )\delta x_k+\displaystyle\sum_{k=1}^{s_2}
E_k(L_\alpha )\delta y_k.
\]
for $\alpha=1,\ldots,s_1.$  Clearly from \eqref{49}  these relations
hold identically.

{F}rom differential equations \eqref{48}, eliminating the Lagrangian
multipliers we obtain equations \eqref{Pa3}. After some
computations we obtain
\begin{equation}\label{488}
\begin{array}{rl}
\dfrac{d}{dt}\left(\dfrac{\partial\,L_0}{\partial\dot{y}_k}-
\displaystyle\sum_{\alpha=1}^{s_1}\dfrac{\partial\,L_\alpha}{\partial\dot{y}_k}
\dfrac{\partial\,L_0}{\partial\dot{x}_\alpha}\right)
-&\left(\dfrac{\partial\,L_0}{\partial{y}_k}-\displaystyle\sum_{\alpha=1}^{s_1}\dfrac{\partial\,L_\alpha}{\partial{y}_k}
\dfrac{\partial\,L_0}{\partial{\dot{x}}_\alpha}\right)+\vspace{0.2cm}\\
&\displaystyle\sum_{\alpha=1}^{s_1}\left(\dfrac{\partial\,L_0}{\partial{x}_\alpha}-
\displaystyle\sum_{\beta=1}^{s_1}\dfrac{\partial\,L_\beta}{\partial{x}_\alpha}
\dfrac{\partial\,L_0}{\partial\dot{x}_\beta}\right)\dfrac{\partial\,L_\alpha}{\partial\dot{y}_k}=0,
\end{array}
\end{equation}
for $ k=1,\ldots,s_2.$

By introducing the function
$\Theta=\left.L_0\right|_{L_1=\ldots=L_{s_1}=0},$ equations
\eqref{488} can be written as
\begin{equation}\label{4488}
\dfrac{d}{dt}\left(\dfrac{\partial\,\Theta}{\partial\dot{y}_k}\right)
-\left(\dfrac{\partial\,\Theta}{\partial{y}_k}\right)+
\displaystyle\sum_{\alpha=1}^{s_1}\left(\dfrac{\partial\,\Theta}{\partial{x}_\alpha}\right)\dfrac{\partial\,L_\alpha}{\partial\dot{y}_k}=0,
\end{equation}
for $k=1,\ldots,s_2.$ Here we consider that
$\dfrac{d}{dt}\left(\dfrac{\partial L_\beta}{\partial
\dot{x}_\alpha}\right)=0,$ for $\alpha,\,\beta=1,\ldots,s_1.$

\smallskip

 We shall study the case
when  equations \eqref{4488} hold identically, i.e. $\Theta=0.$ We choose
\begin{equation}\label{4888}
L_0=\tilde{L}\left(\textbf{x},\textbf{y},\dot{\textbf{x}},\dot{\textbf{y}}\right)-
 \tilde{L}\left(\textbf{x},\textbf{y},\Phi,\dot{\textbf{y}}\right)=\tilde{L}-L^*,
 \end{equation}
 being $\tilde{L}$ the Lagrangian of \eqref{Ra1}.
 Now we establish the relations between equations \eqref{48}  and
 the classical Voronets differential equations with the Lagrangian function
 $L^*=\left.\tilde{L}\right|_{L_1=\ldots= L_{s_1}= 0}.$
The functions $\tilde{L}$ and $L^*$  are determined in such a
way that equations \eqref{Pa3} take place in view of the
equalities
\[ E_{k}\tilde{L}=\displaystyle\sum_{\alpha=1}^{s_1}E_\alpha\tilde{L}
\dfrac{\partial\,L_\alpha}{\partial\dot{y}_k},
\]
and
\[
E_{k}L^*=-\displaystyle\sum_{\alpha=1}^{s_1}\left(-E_k(L_\alpha\,)+
\displaystyle\sum_{\nu=1}^{s_1}E_\nu\,(L_\alpha)\dfrac{\partial\,L_\nu}{\partial\dot{y}_k}\right)
\dfrac{\partial\,\tilde{L}}{\partial\dot{x}_\alpha}-\displaystyle\sum_{\nu=1}^{s_1}E_\nu(L^*)
\dfrac{\partial L_\nu}{\partial\dot{y}_k},
\]
for $ k=1,\ldots,s_2,$ which in view of equalities
$\dfrac{d}{dt}\left(\dfrac{\partial L^*}{\partial
\dot{x}_\nu}\right)=0$ for $\nu=1,\ldots,s_1,$ take the form
\eqref{C48}.
\end{proof}
\begin{proof}[Proof of Proposition \ref{aa1}]
Equations \eqref{C48} describe the motion of the constrained
generalized Voronets systems with Lagrangian $L^*$ and constraints
\eqref{Ch22}. The classical Voronets equations for scleronomic
systems are easy to obtain from \eqref{C48} with
$\Phi_\alpha=\displaystyle\sum_{k=1}^{s_2}a_{\alpha\,k}(\textbf{x},\textbf{y})\dot{y}_k.$
\end{proof}

\smallskip

 Finally by considering Corollary \ref{BBB}  we get that differential equations \eqref{C48}
 describe locally  the motions of any constrained Lagragian systems.

\subsection{Generalized Chaplygin systems}

The   constrained Lagrangian mechanical systems with Lagrangian
$
\tilde{L}=\tilde{L}\left(\textbf{y},\dot{\textbf{x}},
\dot{\textbf{y}}\right),
 $
 and constraints \eqref{Ch2}
is called the {\it Chaplygin mechanical systems}.

\smallskip

The constrained Lagrangian systems
\[
\left(\textsc{Q},\quad\tilde{ L}\left(\textbf{y},\dot{\textbf{x}},
\dot{\textbf{y}}\right),\qquad
\{\dot{x}_\alpha-\Phi_\alpha\left(\textbf{y},\,\dot{\textbf{y}}\right)=0,\quad
\alpha=1,\ldots,s_1\}\right)
\]
is called the {\it generalized Chaplygin systems}. Note that now the Lagrangian do
not depend on $\textbf{x}$ and the constraints  do not depend on $\textbf{x}$ and $\dot{\textbf{x}}.$
So, the  generalized Chaplygin systems are a particular  case of the generalized Voronets system.

\begin{proof}[Proof of Proposition \ref{aa2}]

\smallskip

 To determine the differential equations which describe the behavior of the generalized
Chaplygin systems we apply  Theorem \ref{A1}, with
\[L_0={L}_0\left(\textbf{y},\dot{\textbf{x}},
\dot{\textbf{y}}\right),\quad L_\alpha=\dot{x}_\alpha-\Phi_\alpha\left(\textbf{y} \dot{\textbf{y}}\right),\quad
L_{\beta}=\dot{y}_\beta,
\]
{for} $\alpha=1,\ldots,s_1$ and $\beta=s_1+1,\ldots,s_2$
 and consequently the matrix $W_1$ is given by the formula
 \eqref{45} and
\begin{equation}\label{5000}
\begin{array}{rl}
 A=\Omega_1:=&\left(\begin{array}{ccccccc}
E_1(L_1)&\hdots&E_{s_1}(L_{1})&E_{s_1+1}(L_{1})&\hdots&E_N(L_1)\\
\vdots&\hdots&\vdots&\hdots&\hdots&\vdots\\
E_1(L_{s_1})&\hdots&E_{s_1}(L_{s_1})&E_{s_1+1}(L_{s_1})&\hdots&E_N(L_{s_1})\\
0&\hdots&0&\hdots&0&0\\
\vdots&\hdots&\vdots&\hdots&\hdots&\vdots\\
0&\hdots&0&\hdots&0&0\end{array}
\right)\vspace{0.30cm}\\
=&\left(\begin{array}{ccccccc}
0&\hdots&0&E_{s_1+1}(L_{1})&\hdots&E_N(L_1)\\
\vdots&\hdots&\vdots&\hdots&\hdots&\vdots\\
0&\hdots&0&E_{s_1+1}(L_{s_1})&\hdots&E_N(L_{s_1})\\
0&\hdots&0&\hdots&0&0\\
\vdots&\hdots&\vdots&\hdots&\hdots&\vdots\\
0&\hdots&0&\hdots&0&0\end{array} \right),
\end{array}
\end{equation}
Therefore the differential equations \eqref{210} take the form
\begin{equation}\label{4000}
\begin{array}{rl}
 E_jL_0=&\dfrac{d}{dt}\left(\dfrac{\partial
L_0}{\partial\dot{x}_\alpha}\right)=\dot{\lambda}_j\quad
j=1,\ldots,s_1,\\
 E_{k}L_0=&\displaystyle\sum_{\alpha=1}^{s_1}\left(E_k L_\alpha\, \dfrac{\partial
L_{0}}{\partial\dot{x}_\alpha}+\dot{\lambda}_\alpha
\dfrac{\partial L_\alpha}{\partial\dot{y}_k}\right)\quad
k=1,\ldots,s_2.
\end{array}
\end{equation}
The transpositional relations are
\begin{equation}\label{449}
\begin{array}{rl}
&\delta\dfrac{dx_\alpha}{dt}-\dfrac{d}{dt}\delta\,x_\alpha=\displaystyle\sum_{k=1}^{s_2}
E_k(L_\alpha )\delta y_k,\quad \alpha=1,\ldots,s_1,\\
&\delta\dfrac{dy_m}{dt}-\dfrac{d}{dt}\delta\,y_m=0,\quad
m=1,\ldots,s_2.
\end{array}
\end{equation}
By excluding the Lagrangian multipliers from \eqref{4000} we obtain the equations
\[ E_{k}L_0=\displaystyle\sum_{\alpha=1}^{s_1}\left(E_k(L_\alpha
)\dfrac{\partial L_0}{\partial\dot{x}_\alpha}+\dfrac{d}{dt}\left(\dfrac{\partial
L_0}{\partial\dot{x}_\alpha}\right)\dfrac{\partial L_\alpha}{\partial\dot{y}_k}\right),\] for $
k=1,\ldots,s_2.$

\smallskip

In this case equations \eqref{4888} take the form
\begin{equation}\label{44488}
\dfrac{d}{dt}\left(\dfrac{\partial\,\Theta}{\partial\dot{y}_k}\right)
-\left(\dfrac{\partial\,\Theta}{\partial{y}_k}\right)=0,
\end{equation}
Analogously to the Voronets case we study the subcase when
$\Theta=0.$ We choose
$L_0=\tilde{L}\left(\textbf{y},\dot{\textbf{x}},
\dot{\textbf{y}}\right)-\tilde{L}\left(\textbf{y},\Phi,
\dot{\textbf{y}}\right):=\tilde{L}-L^*.$ We assume that the functions
$\tilde{L}$ and $L^*$ are such that
\begin{equation}\label{001}
E_{k}L^*=-\displaystyle\sum_{\alpha=1}^{s_1}E_k(L_\alpha
)\dfrac{\partial
\tilde{L}}{\partial\dot{x}_\alpha}\Psi_\alpha,\end{equation} where
$\Psi_\alpha=\left.\dfrac{\partial
\tilde{L}}{\partial\dot{x}_\alpha}\right|_{L_1=\ldots=L_{s_1}=0}$
and
\[
E_k(\tilde{L})=\displaystyle\sum_{\alpha=1}^{s_1}\dfrac{d}{dt}\left(\dfrac{\partial
\tilde{L}}{\partial\dot{x}_\alpha}\right)\dfrac{\partial
L_\alpha}{\partial\dot{y}_k},
\]
for $ k=1,\ldots,s_2.$

\smallskip

By inserting $\dot{x}_j=\displaystyle\sum_{k=1}^{s_2}a_{j\,k}(\textbf{y})\dot{y}_k,\quad
j=1,\ldots,s_1,$ into equations \eqref{001} we obtain system \eqref{Chh2}.
Consequently system \eqref{001} is an extension of the
classical Chaplygin equations when the constraints
are nonlinear.
\end{proof}

\smallskip

For the generalized Chaplygin systems the Lagrangian $L$ takes the
form
\begin{equation}\label{bbb}
L=\tilde{L}(\textbf{y},\dot{\textbf{x}},\dot{\textbf{y}})-
\tilde{L}(\textbf{y},\Phi,\dot{\textbf{y}})-\displaystyle\sum_{j=1}^{s_1}\left(\dfrac{\partial
L^{*}}{\partial\dot{x}_j}+C_j\right)\left(\dot{x}_j-\Phi_j(\textbf{y},\dot{\textbf{y}})\right)-
\displaystyle\sum_{j=}^{s_2}\lambda^0_j\dot{y}_j,
\end{equation}
for $j=1,\ldots,s_1$
where the constants $C_j$ for
$j=1,\ldots,s_1$ are arbitrary. Indeed, from \eqref{4000} follows that
\[\lambda_j=\dfrac{\partial
L_0}{\partial\dot{x}_j}+C_j=\dfrac{\partial
L^{*}}{\partial\dot{x}_j}+C_j.\] By inserting in \eqref{21} $L_0=\tilde{L}-L^*$ and $\lambda_j$ for $j=1,\ldots,s_1$
we obtain function $L$ of \eqref{bbb}.

\smallskip

We note that Vorones and Chaplygin equations with nonlinear constraints in the velocity
was also obtained by Rumiansev and Sumbatov $($see \cite{Rumiansev,Sumbatov}$)$.

\smallskip

\textbf{Example 4.} We shall illustrate the above results in the
following example.

In the Appel's and Hamel's investigations the following mechanical
system was analyzed. A weight of mass $m$ hangs on a thread which
passes around the  pulleys and is wound round the drum of radius
$a$. The drum is  fixed to a wheel of radius $b$ which rolls
without sliding on a  horizontal plane, touching it at the point
$B$ with the coordinates $(x_B,\,y_B)$. The legs of the
 frame that support the pulleys and keep the plane of the wheel
vertical slide on the horizontal plane without  friction. Let
$\theta$ be the angle between the plane of the wheel and the $Ox$
axis; $\varphi$ the angle of the rotation of the wheel in its own
plane; and $(x,y,z)$ the coordinates of the mass $m.$ Clearly,
$$\dot{z}=b\dot\varphi,\quad b>0 .$$
The coordinates of the point $B$ and the coordinates of the mass
are related as follows (see page 223 of \cite{NF} for a picture)
$$x=x_B+\rho\cos\theta,\quad y=y_B+\rho\sin\theta .$$
The condition of rolling without sliding leads to the equations of
nonholonomic constraints:
$$\dot{x}_B=a\cos\theta\dot{\varphi},\quad
\dot{y}_B=a\sin\theta\dot{\varphi}\quad b>0 .
$$
We observe that the constraints $\dot{z}=b\dot{\varphi}$ admits
the representation
\[\dot{z}=\dfrac{b}{a}\sqrt{\dot{x}^2+\dot{y}^2-\rho^2\dot\theta^2}.\]
 Denoting by $m_1,\,A$ and $C$ the
mass and the  moments of inertia of the wheel and neglecting the
mass of the frame, we obtain the following expression for the
Lagrangian function
\[
\tilde{L}
=\dfrac{m+m_1}{2}\left(\dot{x}^2+\dot{y}^2\right)+\dfrac{m}{2}\dot{z}^2+m_1\rho\dot{\theta}\left(\sin
\theta\dot{x}-\cos\theta\dot{y}\right)+\dfrac{A+m_1\rho^2}{2}\dot{\theta}^2+\dfrac{C}{2}\dot{\varphi}^2
-mgz.
\]
The equations of the constraints are
\[
\dot{x}-a\cos\theta\dot\varphi+\rho\sin\theta\dot\theta=0
,\quad\dot{y}-a\sin\theta\dot\varphi-\rho\cos\theta\dot\theta=0
,\quad\dot{z}-b\dot\varphi =0,
\]
Now we shall study the motion of this constrained Lagrangian
 in the coordinates
\[
x_1=x,\,x_2=y,\,x_3=\dot{\varphi}, y_1=\theta,\,y_2=z.\] i.e., we
shall study the nonholonomic system with Lagrangian
\[\begin{array}{rl}
\tilde{L}=&\tilde{L}\left(y_1,\,y_2,\,\dot{x}_1,\,\dot{x}_2,\,\dot{x}_3,\,\dot{y}_1,\,\dot{y}_2\right)\vspace{0.20cm}\\
=&\dfrac{m+m_1}{2}\left(\dot{x}^2_1+\dot{x}^2_2\right)+\dfrac{C}{2}\dot{x}^2_3+\dfrac{J}{2}\dot{y}^2_1+
\dfrac{m}{2}\dot{y}^2_2 +m_1\rho\dot{y_1}\left(\sin\,
y_1\dot{x}_1-\cos\,y_1\dot{x}_2\right)-\dfrac{mg}{b}y_2,
\end{array}
\]
and with the constraints
\[\begin{array}{rl}
l_1=&\dot{x}_1-\dfrac{a}{b}\,\dot{y}_2\cos
y_1-\rho\dot{y}_1\sin\,y_1=0,\vspace{0.2cm}\\
 l_2=&\dot{x}_2-\dfrac{a}{b}\dot{y}_2\sin
y_1+\rho\dot{y}_1\cos\,y_1=0 ,\vspace{0.2cm}\\
l_3=&\dot{x}_3-\dfrac{1}{b}\dot{y}_2 =0.
\end{array}\]
Thus we have a classical Chaplygin system.  To determine
differential equations \eqref{001} and the
transpositional relations \eqref{449}  we define the functions:
\[\begin{array}{rl}
L^*=-&\tilde{L}|_{l_1=l_2=l_3=0}=
\dfrac{m(a^2+b^2)m+a^2m_1+C}{2b^2}\dot{y}^2_2+\dfrac{m
\rho^2+J}{2}\dot{y}^2_1-\dfrac{mg}{b}y_2,\vspace{0.20cm}\\
L_1=&l_1,\quad L_2=l_2,\quad L_3=l_3,\quad L_4=\dot{y}_1,\quad
L_5=\dot{y}_2.
\end{array}
\]
After some computations we obtain that the matrix $A$ (see formulae
\eqref{5000}) in this case
 becomes
\[
A=\left(
\begin{array}{cccccc}
0&0&0&-\dfrac{a}{b}\dot{y}_2\sin\,y_1&\dfrac{a}{b}\dot{y}_1\sin\,y_1\vspace{0.20cm}\\
0&0&0&\dfrac{a}{b}\dot{y}_2\cos\,y_1&-\dfrac{a}{b}\dot{y}_1\cos\,y_1\vspace{0.20cm}\\
0&0&0&0&0\\
0&0&0&0&0\\
0&0&0&0&0
\end{array} \right),
\]
thus  differential equations \eqref{001}  take the form
\[\begin{array}{rl}
&\left(m\rho^2+J\right)\ddot{y}_1+\dfrac{a\rho m}{b}\dot{y}_1\dot{y}_2=0,\vspace{0.20cm}\\
&\left((m+m_1)a^2+mb^2\right)\ddot{y}_2-{mab\rho}\dot{y}^2_1=-mgb.\end{array}\]
Assuming that $(m+2m_1)\rho^2+J\ne 0$ and by considering the
existence of the first integrals
\[\begin{array}{rl}
C_2=&\dot{y}_1\exp{\left(-\dfrac{a\varrho
m y_2}{b\left(m\rho^2+J\right)}\right)},\vspace{0.2cm}\\
h=&\dfrac{\left((m+m_1)a^2+mb^2\right)}{2}\dot{y}^2_2+\dfrac{b^2\left(m\rho^2+J\right)}{2}\dot{y}^2_1+mgby_2,\\
\end{array}
\]
after the integration of these first integrals we obtain
\[\begin{array}{rl}
&\displaystyle\int\dfrac{\sqrt{(m+m_1)a^2+mb^2}dy_2}{\sqrt{2h-2mgby_2-
{b^2\left(m\rho^2+J\right)}C_3\exp{\left(\dfrac{a\rho\,my_2}{b{m\rho^2+J}}\right)}}}=t+C_1,\vspace{0.30cm}\\
&y_1(t)=C_3+C_2\displaystyle\int\exp{\left(2\dfrac{a\rho\,my_2(t)}{b{m\rho^2+J}}\right)}dt.
\end{array}
\]
Consequently, if $\rho=0$ then
\[y_1=C_3+C_2t,\qquad \displaystyle\int\dfrac{\sqrt{(m+m_1)a^2+mb^2}dy_2}{\sqrt{2h-2mgby_2-
{J}C_3}}=t+C_1.
\]
 Hamel in \cite{Hamel} neglect the mass of the wheel
($ m_1=J=C=0$). Under these conditions the previous equations
become
\[\begin{array}{rl}
&\rho^2\ddot{y}_1+\dfrac{a\rho }{b}\dot{y}_1\dot{y}_2=0,\vspace{0.20cm}\\
&(a^2+b^2)\ddot{y}_2-ab\rho\dot{y}^2_1=-gb\end{array}\] Appell and
Hamel obtained the example of nonholonomic system with nonlinear
constraints by means of the passage to the limit $\rho\to 0.$
However, as a result of this limiting process, the order of the
system of differential equations is reduced, i.e., they become
degenerate. In \cite{NF} the authors study the motion of the
nondegenerate system for $\rho>0$ and $\rho<0.$  From these studies
it follows that the motion of the nondegenerate system ($\rho\ne 0$)
and degenerate system ($\rho\to 0$) differ essentially. Thus the
Appell-Hamel example with nonlinear constraints is incorrect.

 The transpositional relations \eqref{449}
become
\[\begin{array}{rl}
\delta\dfrac{dx_1}{dt}-\dfrac{d\delta\,x_1}{dt}=&\dfrac{a}{b}\sin\,y_1
\left(\dfrac{dy_1}{dt}\delta{y_2}-\dfrac{dy_2}{dt}\delta{y_1}\right),\vspace{0.30cm}\\
\delta\dfrac{dx_2}{dt}-\dfrac{d\delta\,x_2}{dt}=&\dfrac{a}{b}\cos\,y_1
\left((\dfrac{dy_1}{dt}\delta{y_2}-\dfrac{dy_2}{dt}\delta{y_1}\right),\vspace{0.30cm}\\
\delta\dfrac{dx_3}{dt}-\dfrac{d\delta\,x_3}{dt}=&0,\quad
\delta\dfrac{dy_1}{dt}-\dfrac{d\delta\,y_1}{dt}=0,\quad
\delta\dfrac{dy_2}{dt}-\dfrac{d\delta\,y_2}{dt}=0.
\end{array}
\]
Clearly these relations are independent of $\varrho,\,A,\,C$ and
$m_1.$

\section{Consequences of Theorems \ref{A1} and \ref{A2} and  the proof of Corollary \ref{CC}.}

We observe the following important aspects from  Theorems
\ref{A1} and \ref{A2}.

\smallskip

(I) Conjecture \ref{CC11} is supported by the following facts. (a)
As a general rule the constraints studied in classical mechanics are
linear in the velocities. However Appell and Hamel  in 1911,
considered an artificial example with a constraint nonlinear in the
velocity . As it follows from \cite{NF} (see example 4) this
constraint does not exist in the Newtonian mechanics.

\smallskip

(b) The idea developed for some authors (see for instance
\cite{Birkhoff}) to construct a theory in Newtonian mechanics, by
allowing that the field of force depends on the acceleration, i.e.
function of $\ddot{\textbf{x}}$ as well as of the position
$\textbf{x},$ velocity $\dot{\textbf{x}},$ and the time $t$ is
inconsistent with one of the fundamental postulates of the Newtonian
mechanics: when two forces act simultaneously on a particle the
effect is the same as that of a single force equal to the resultant
of both forces (for more details see \cite{Pars} pages 11--12).
Consequently the forces depending on the acceleration are not
admissible in Newtonian dynamics. This does not preclude their
appearance in electrodynamics, where this postulate does not hold.

\smallskip
(c) Let $T$ be the kinetic energy of the constrained Lagrangian
systems. We consider the generalization of the Newton law: {\it the
acceleration $($see \cite{Singe,Oliva}$)$
\[\dfrac{d}{dt}\dfrac{\partial T
}{\partial\dot{\textbf{x}}}-\dfrac{\partial T
}{\partial{\textbf{x}}}\] is equal to the force} $\textbf{F}.$ Then in the
differential equations \eqref{210} with $L_0=T$  we obtain that
the field of force $\textbf{F}$ generated by the constraints is
\[\textbf{F}=\left(W^{-1}_1\Omega_1\right)^T
\dfrac{\partial{T}}{\partial{\dot{\textbf{x}}}}+
W^T_1\dfrac{d}{dt}\lambda:=\textbf{F}_1+\textbf{F}_2.
\]
The field of force
$\textbf{F}_2=W^T_1\dfrac{d}{dt}\lambda=\left(F_{21},\ldots,F_{2N}\right)$
is called the {\it reaction force of the constraints}.
  What is the meaning of the force
\begin{equation}\label{F1}
\textbf{F}_1=\left(W^{-1}_1\Omega_1\right)^T
\dfrac{\partial{T}}{\partial{\dot{\textbf{x}}}}\,?
\end{equation}

\smallskip

 If the constraints are nonlinear in the
velocity, then $\textbf{F}_1$ depends on $\ddot{\textbf{x}}.$
Consequently in Newtonian mechanics does not exist a such field of
force. Therefore, the existence of nonlinear constraints in the velocity
 and the meaning of force $\textbf{F}_1$
must be sought outside of the Newtonian model.

For example, for the Appel-Hamel constrained Lagrangian  systems
studied in the previous subsection we have that
\[
\textbf{F}_1=\left(-\dfrac{a^2\dot{x}}{\dot{x}^2+
\dot{y}^2}(\dot{x}\ddot{y}-\dot{y}\ddot{x}),\,
\dfrac{a^2\dot{y}}{\dot{x}^2+\dot{y}^2}(\dot{x}\ddot{y}-\dot{y}\ddot{x}),\,0\right).
\]
For the generalized Voronets systems and locally for any
nonholonomic constrained Lagrangian systems
from the equations \eqref{48} we obtain that the field of force
$\textbf{F}_1$ has the following components
\begin{equation}\label{Kk01}
\begin{array}{rl}
F_{k\,1}=&\displaystyle\sum_{\alpha=1}^{s_1}E_k L_\alpha\,
\dfrac{\partial
L_{0}}{\partial\dot{x}_\alpha}\\
=&\displaystyle\sum_{j=1}^N
\sum_{\alpha=1}^{s_1}\left(\dfrac{\partial^2L_\alpha}{\partial\dot{x}_k\dot{x}_j}\dfrac{\partial
L_0}{\partial\dot{x}_\alpha}\ddot{x}_j+\dfrac{\partial^2L_\alpha}{\partial\dot{x}_k\partial{x}_j}\dfrac{\partial
L_0}{\partial\dot{x}_\alpha}\dot{x}_j\right)+
\displaystyle\sum_{\alpha=1}^{s_1}\dfrac{\partial^2L_\alpha}{\partial\dot{x}_k\partial
t}\dfrac{\partial L_0 }{\partial\dot{x}_\alpha},\quad\mbox{for}\quad
k=1\ldots,N,\quad s_1=M.
\end{array}
\end{equation}
consequently such field of force does not exist in Newtonian
mechanics if the constraints are nonlinear in the velocity.

\smallskip

 (II)  Equations \eqref{210} can be rewritten in the form
\begin{equation}\label{K01}
G\ddot{\textbf{x}}+\textbf{f}(t,\textbf{x},\dot{\textbf{x}})=0,
\end{equation}
where $G=G(t,\textbf{x},\dot{\textbf{x}})$ is the matrix $\left(
G_{j,k}\right)$ given by
\[
G_{jk}=\dfrac{\partial^2L_0}{\partial\dot{x}_j\partial\dot{x}_k}-
\displaystyle\sum_{n=1}^N\dfrac{\partial
A_{nk}}{\partial\ddot{x}_j}\dfrac{\partial
L_0}{\partial\dot{x}_n},\quad j,k=1,\ldots,N,
\] and $\textbf{f}(t,\textbf{x},\dot{\textbf{x}})$ is a convenient
vector function. If $\det G\ne 0$ then equation \eqref{K01} can be
solved with respect to $\ddot{\textbf{x}}.$ This implies, in
particular that the motion of the mechanical system at time
$\overline{t}\in[t_0,\,t_1]$ is uniquely determined, i.e. the
{\it principle of determinacy} (see for instance \cite{Arnold}) holds for the mechanical
systems with equation of motion given in \eqref{210}.

\smallskip

 In particular for the Appel-Hamel constrained Lagrangian
 systems we have (see formula \eqref{315}) that
 \[\begin{array}{rl}
\textbf{x}=&\left(x,\,y,\,z\right)^T,\quad
\textbf{f}=\left(\dfrac{a\dot{x}}{\sqrt{\dot{x}^2+\dot{y}^2}}\dot{\lambda},
\,\dfrac{a\dot{y}}{\sqrt{\dot{x}^2+\dot{y}^2}}\dot{\lambda},\,g-\dot{\lambda}\right)^T\vspace{0.20cm}\\
G=&\left(
\begin{array}{cccc}
1+\dfrac{a^2\dot{y}^2}{\dot{x}^2+\dot{y}^2}&-\dfrac{a^2\dot{x}\dot{y}}{\dot{x}^2+\dot{y}^2}&0\\
-\dfrac{a^2\dot{x}\dot{y}}{\dot{x}^2+\dot{y}^2}&1+\dfrac{a^2\dot{x}^2}{\dot{x}^2+\dot{y}^2}&0\\
0&0&1\end{array} \right) ,\quad |G|=1+a^2.
\end{array}
\]
So in the Appel--Hamel system the principle of determinacy  holds.

\smallskip

(III)\begin{proof}[Proof of Corollary  \ref{CC}] {From} Theorems
\ref{A1} and \ref{A2} (see formulas \eqref{211} and \eqref{0211})
and from  all examples which we
 gave in the previous sections we see that are examples with zero transpositional relations
 and examples where all they are not zero. By contrasting the MVM with the Lagrangian mechanics
 we obtain that for the unconstrained Lagrangian systems the
 transpositional relations are always zero. Thus we have the proof of the
corollary.
\end{proof}

\section*{Acknowledgements}

The first author is partially supported by a MINECO/FEDER grant
number MTM2009-03437, an AGAUR grant number 2009SGR-410, ICREA
Academia and FPZ--PEOPLE--2012--IRSES--316338 and 318999. The second
author was partly supported by the Spanish Ministry of Education
through projects TSI2007-65406-C03-01 ``E-AEGIS" and Consolider
CSD2007-00004 ``ARES".

\end{document}